\documentclass[a4paper,11pt,leqno]{article}
\usepackage{mathtools}
\usepackage{geometry}
\usepackage{layout}
\usepackage{array}
\usepackage{cite}
\usepackage{hyperref}
\usepackage{amsmath,amssymb,amsthm}
\usepackage[french,english]{babel}
\usepackage[applemac]{inputenc}
\usepackage[cyr]{aeguill}
\usepackage{esint}
\usepackage[usenames,dvipsnames]{color}
\usepackage{xargs}
\usepackage{enumitem}
\setlist{nosep} 

\def\XXint#1#2#3{{\setbox0=\hbox{$#1{#2#3}{\int}$}
     \vcenter{\hbox{$#2#3$}}\kern-.5\wd0}}

\newcommand{\N}{\mathbb{N}}
\newcommand{\Z}{\mathbb{Z}}

\newcommand{\R}{\mathbb{R}}
\newcommand{\C}{\mathbb{C}}

\newtheorem{theo}{Theorem}
\newtheorem{prop}{Proposition}[section]
\newtheorem{lem}[prop]{Lemma}
\newtheorem{coro}[prop]{Corollary}
\newtheorem{remark}[prop]{Remark}
\newtheorem{defi}[prop]{Definition}

\theoremstyle{plain}

\numberwithin{equation}{section}

\def\t0{\rightarrow 0} 
\def\ti{\rightarrow \infini} 

\newcommand{\f}{\frac}
\newcommand{\infini}{\infty}
\newcommand{\ep}{\varepsilon}

\newcommand{\hal}{\frac{1}{2}}

\def\div{\mathrm{div }} 
\def\1{\mathbf{1}} 

\def \mc{\mathcal}

\def \ep{\epsilon}
\renewcommand{\epsilon}{\varepsilon}
\def \({\left(}
\def \){\right)}

\def\Rd{\R^d} 
\def\meseq{\mu_{eq}} 
\def \ZNbeta{Z_{N,\beta}} 
\def \carr{C} 
\def \Nn{\mc{N}}
\def \D{\mc{D}}

\def\yg{|y|^\gamma}
\newcommand{\drd}{\delta_{\Rd}}
\def \W{\mathbb{W}} 
\def \WN{\mc{W}_N} 

\def\config{\mathcal{X}} 
\def\Lip{\mathrm{Lip}_1} 
\def \Lploc{L^p_{\mathrm{loc}}}

\def \probas{\mathcal{P}}
\def \Pelec{P^{\rm{elec}}} 

\def\P{\mathbb{P}} 
\def \Pst{P} 

\def \PNbeta{\P_{N, \beta}} 
\def \PgN2{\mathbf{P}_{N,2}} 

\def \sineb{\text{Sine}_{\beta}} 

\def\Esp{\mathbf{E}} 
\def \E{\Esp}

\def \Ent{\mathrm{Ent}}   
\def \ERS{\mathsf{ent}} 

\def \Poisson{\mathbf{\Pi}}
\def \B{\mathbf{B}} 

\def \fbeta{\mathcal{F}_{\beta}} 

\def \cds{c_{d,s}} 
\def \A{\mathsf{Elec}} 
\def \conf{\rm{Conf}} 
\def \Eloc{E^{\rm{loc}}}
\def \Hloc{\Phi^{\rm{loc}}}

\def \dist{d}

\def \dist{\mathrm{dist}}
\def\c{c_{d,s}}

\def\g{g}

\def\nab{\nabla}

\newcommandx \So[2][1=e,2=M]{\mathcal{S}^{#2,#1,\epsilon}_{R, \eta,-}}
\renewcommandx \H[2][1=e, 2=M]{H^{#2,#1,\epsilon}_{R, \eta}}

\def \Int{\mathrm{Int}}

\def \Pelec{\Pst^{\rm{elec}}}

\DeclarePairedDelimiter\ceil{\lceil}{\rceil}
\def \llceil{\left\lceil}
\def \rrceil{\right\rceil}
\DeclarePairedDelimiter\floor{\lfloor}{\rfloor}

\def \Welec{\mathbb{W}^{\rm{elec}}}
\def \Wper{\mathbb{W}^{\rm{BS}}}
\def \Wint{\mathbb{W}^{\mathrm{int}}}

\def \Sp{\mathcal{S}_{R, \eta}^{M, \epsilon}}

\def \Hint{\mathcal{H}^{\rm{int}}}
\def \C{\mathcal{C}}

\def \hcarr{\hat{C}}
\def \SR{S_R}
\def \hB{\hat{B}}
 
\def \H{H}
\def \wN{W_N}

\def \Dlog{\mathcal{D}^{\rm{log}}}
\def \Clog{C^{\rm{log}}}
\def \fbetai{\mc{F}'_{\beta}}
 \def \Escr{E^{\rm{scr}}}
 \def \Cscr{\C^{\rm{scr}}}
\def \Cper{\C^{\rm{per}}}
\def \Eper{E^{\rm{per}}}
\def \Pper{P^{\rm{per}}}

\def \Pbeta{P_{\beta}}

\def \TV{\rm{TV}}
\def \E{\mc{E}}

\def \Pav{P^{\rm{av}}}

\def \bC{\bar{C}}

\def \Emod{E^{\rm{mod}}}

\def \hC{\hat{C}}

\def \c2d{c^{\rm{log}}_{d=2}}

\def \Gper{G^{\rm{per}}}

\def \T{\mathbb{T}}

\def \Elocn{E^{\mathrm{loc}, n}}

\renewcommand \Sp{\mathcal{S}^{M,\epsilon}_{R, \eta}}

\def \Inte{\mathrm{Int}}

\def \Spn{\mathcal{S}_{R_n, \eta}^{M, \epsilon}}
\def \Pmod{P^{\rm{mod}}}

\begin{document}
\title{Logarithmic, Coulomb and Riesz energy of point processes}
\author{Thomas Leblé}
\maketitle
\begin{abstract}
We define a notion of logarithmic, Coulomb and Riesz interactions in any dimension for random systems of infinite charged point configurations with a uniform background of opposite sign. We connect this interaction energy with the “renormalized energy” studied by Serfaty \textit{et al.}, which appears in the free energy functional governing the microscopic behavior of logarithmic, Coulomb and Riesz gases. Minimizers of this functional include the Sine-beta processes in the one-dimensional Log-gas case. Using our explicit expression (inspired by the work of Borodin-Serfaty) we prove their convergence to the Poisson process in the high-temperature limit as well as a crystallization result in the low-temperature limit for one-dimensional systems.
\end{abstract}
\section{Introduction}
\subsection{General setting}
\paragraph{Logarithmic, Coulomb and Riesz interactions.}
We consider a system of points (which we can think of as being point particles carrying a positive unit charge) in the Euclidean space $\Rd$ interacting via logarithmic, Coulomb or Riesz pairwise interactions
\begin{equation} \label{wlog}
g(x)=-\log |x| , \quad \text{in dimension } d=1,
\end{equation} 
\begin{equation}
\label{wlog2d}
g(x) = - \log |x| ,\quad  \text{in dimension } d=2,
\end{equation} 
or in general dimension
\begin{equation}\label{kernel}  
g(x)=\frac{1}{|x|^{s}}, \quad \max(0, d-2) \leq s<d.
\end{equation}  

Cases \eqref{wlog} and \eqref{wlog2d} are known as one- and two-dimensional log-gases, and we will refer to them as the “logarithmic cases”. One-dimensional log-gases have been extensively studied for their connection with important random matrix models known as the $\beta$-ensembles (see \cite{forrester}). The two-dimensional log-gas is known in the physics literature as a \textit{two-dimensional one-component plasma} (see \cite{alastueyjancovici}) and can also model non-Hermitian random matrices such as the Ginibre ensemble \cite{Ginibre}. The cases \eqref{kernel} correspond to higher-dimensional Coulomb gases (if $s=d-2$) or Riesz gases. 

The statistical mechanics of $N$ points $(x_1, \dots, x_N)$ interacting pairwise via $g$ under a confining potential $V$ at inverse temperature $\beta \in (0, + \infty)$ is given by the canonical Gibbs measure 
\[
d\PNbeta(x_1, \dots, x_N) := \frac{1}{\ZNbeta} e^{- \beta \left( \sum_{i \neq j} g(x_i-x_j) + N \sum_{i=1}^N V(x_i)\right)} dx_1 \dots dx_N,
\]
where $\ZNbeta$ is a normalizing constant. 
The \textit{macroscopic} behavior i.e. the behavior of the empirical measure $\mu_N := \frac{1}{N} \sum_{i=1}^N \delta_{x_i}$ is well studied in the limit $N \ti$, see e.g. \cite{serfatyZur} and the references therein. The limiting macroscopic arrangement is described by the equilibrium measure $\meseq$, which is a probability measure on $\Rd$, depending on $V$, with compact support $\Sigma$, such that $\{\mu_N\}_N$ converges weakly to $\meseq$, $\PNbeta$-a.s.

In order to study the microscopic arrangement of the particles, Sandier-Serfaty have derived in \cite{SS2d} (see also \cite{SS1d}, \cite{RougSer}, \cite{PetSer}) a \textit{second-order} energy functional $\wN$ which governs the fluctuations around $\meseq$, together with an object $\W$ defined on infinite point configurations which is the limit of $\wN$ as $N \ti$ in the sense of $\Gamma$-convergence of functionals (see \cite{serfatyZur}). Let $x_i' := N^{1/d} x_i$, $\nu'_N := \sum_{i=1}^N \delta_{x'_i}$ and let $\meseq'$ be the push-forward of $\meseq$ by $x \mapsto N^{1/d}x$. The “blown-up” point configuration $\nu'_N$ encodes the position of particles at the microscopic (inter-particle) scale $N^{-1/d}$. We let 
\begin{equation} \label{def:WN}
\wN(x_1, \dots, x_N) := \frac{1}{N} \iint_{\triangle^c} g(x-y) (d\nu'_N - d\meseq'(x))\otimes(d\nu'_N - d\meseq'(y)),
\end{equation}
where $\triangle$ is the diagonal. It can be seen that $\wN$ computes (up to the factor $\frac{1}{N}$)  the Coulomb/Riesz interaction of the electric system made of the finite (charged) point configuration $\nu'_N$ and of a negatively charged background of density $d\meseq'$, with itself, without the infinite self-interactions of the charges because the diagonal $\triangle$ is excluded. 

An integration by parts shows that $\wN$ may be re-written with the help of the associated electric field $\Eloc := \nabla g * (\nu'_N - \meseq')$, whose norm is computed in a \textit{renormalized} fashion to take care of the singularities around each charge (we will come back to this procedure in Section \ref{sec:WSS}).

Following the same procedure, a \textit{renormalized energy} functional $\mc{W}$ is defined on the space of electric fields corresponding to infinite point configurations together with a uniform background of intensity $1$. If $\C$ is a point configuration, its Coulomb/Riesz energy $\W(\C)$ is then defined as
\begin{equation*}
 \W(\C) := \inf_E \mc{W}(E),
\end{equation*}
where the infimum is taken among the set of electric fields $E$ which are compatible with $\C$. Finally if $\Pst$ is a random point process (a probability measure on point configurations)  its energy is defined by 
\begin{equation*}
\Welec (\Pst) := \Esp_{P}[ \W ].
\end{equation*} 
We refer to Section \ref{sec:WSS} for more details. The superscript “elec” is added by us and refers to this “electric” approach to the definition of a Coulomb/Riesz energy.

\paragraph{Free energy at microscopic scale.}
In \cite{LebSer} S. Serfaty and the author have obtained a second-order (or \textit{process level}) large deviation principle concerning the average microscopic behaviour of the particles under the canonical Gibbs measure at inverse temperature $\beta~\in~(0, + \infty)$. This behaviour is characterized by a certain random point process $\Pst$ (the law of a random point configuration) and it amounts to minimizing a  free energy functional of the form
\begin{equation} \label{def:fbeta}
\fbeta(\Pst) := \beta \Welec(P) + \ERS[\Pst|\Poisson]
\end{equation}
on the space $\probas_{s,1}(\config)$ of translation-invariant random point processes whose mean density of points is~$1$. The term $\ERS[\Pst|\Poisson]$ denotes the \textit{specific relative entropy} of $\Pst$ with respect to the Poisson point process of intensity $1$ on $\Rd$, it is the infinite-volume analogue of the usual relative entropy.

From a physics perspective, knowing the minimizers of the free energy and how they behave as $\beta$ varies allows one to retrieve some of the thermodynamic properties of the physical system at the microscopic scale (e.g. the existence of phase transitions). From the random matrix theory point of view, it was proven in \cite{LebSer} that an important family of point processes governing the microscopic behavior of eigenvalues, namely the $\sineb$ processes of Valko-Virag (see \cite{vv}), minimizes $\fbeta$ for $\beta > 0$. Hence it would be very useful to have information on $\fbeta$, its level sets and its minimizers (depending on $\beta$). A drawback of the free energy $\fbeta$ is that computing it explicitly is hard. The energy term in particular is difficult to evaluate and except for the case of periodic configurations (for which exact formulas hold) no value of $\W$ is known and the mere finiteness of $\W(\C)$ for a given point configuration $\C$ is unclear in general (see however \cite{gesandier} for some criteria). On the other hand, level sets of $\W$ are easily seen to be degenerate because small perturbations of any given configuration $\C$ typically do not change its energy. Concerning both issues, it turns out to be helpful to look for a definition of the energy directly at the level of stationary random point processes instead of averaging the energy computed configuration-wise.

\paragraph{The approach of Borodin-Serfaty.}
In \cite{BorSer} a related notion of a renormalized energy for random point processes was introduced in the logarithmic setting \eqref{wlog} and \eqref{wlog2d}. Given a stationary random point process $\Pst$, Borodin and Serfaty proceed by periodizing the point process induced in a square (or interval) of sidelength $R$ and computing its renormalized energy by the mean of explicit formulas valid in the periodic setting. If $(a_1, \dots, a_N)$ is a configuration in a torus $\T$ of volume $N$ in $\Rd$, the associated periodic point configuration has an energy
$$
\W(a_1, \dots, a_N) = \frac{\cds^2}{N} \sum_{i \neq j} \Gper(a_i -a_j) + \cds^2 \lim_{x \t0} \left(\Gper(x) + \frac{\log(x)}{\cds}\right),
$$
where $\cds$ is a constant and $\Gper$ is a certain periodic Green function which has a logarithmic singularity at $0$. Taking the expectation under $P$, using an expansion of $\Gper$ and sending $N \ti$ they obtain an energy $\Wper(\Pst)$ (our notation) which may be written, up to an additive constant, as 
\begin{equation} \label{def:Wper}
\Wper(\Pst) = \int_{\Rd} -\log |v| (\rho_{2,\Pst}(v) - 1) dv,
\end{equation}
where $\rho_{2, \Pst}$ denotes the two-point correlation function of $P$, which can be seen as a function of one variable by stationarity (we abuse notation and let $\rho_{2, \Pst}(v) := \rho_{2, \Pst}(0, v)$). Using this explicit expression in terms of $\rho_{2,P}$, they are able to compute the energy $\Wper$ for some specific point processes (e.g. the $\sineb$ processes for $\beta = 1, 2, 4$, and the Ginibre point process) as well as to solve minimization problems (the minimization of $\Wper$ over a large class of determinantal point processes). However no general rigorous connection is drawn between $\Wper$ and the electric definition $\Welec$ which derives from the energy functional $\WN$. Moreover the formulas of \cite{BorSer} only apply to random point processes for which $\rho_{2,P}(x,y) - 1$ decays fast enough as $|x-y| \ti$. The approach of the present paper is strongly inspired by the one of \cite{BorSer} and is an attempt to give a partial connection between $\Wper$ and $\Welec$.

\subsection{Main results}
The purpose of this paper is twofold. First we introduce an energy $\Wint$ (“int” as “intrinsic”) defined on $\probas_{s,1}(\config)$ which is expressed only in terms of $g$ and of the two-point correlation function, and we connect $\Wint$ with $\Welec$. In a second part we use $\Wint$ to handle the energy term in $\fbeta$, which allows us to describe the behavior of minimizers of $\fbeta$ in the limiting cases $\beta \t0$ (in any dimension) and $\beta \ti$ (in dimension 1). 

\paragraph{New definition of the Coulomb/Riesz energy.}
If $P$ is a stationary random point process of intensity $1$ and $\rho_{2,P}$ denotes its two-point correlation function, we define in Section \ref{sec:Wint} its “intrinsic” energy (with respect to the interaction $g$) as 
$$
\Wint(P) := \liminf_{R \ti} \f{1}{R^d} \iint_{\carr_R^2 \backslash \triangle} g(x,y) (\rho_2(x,y)-1) dxdy,$$
where $\carr_R$ is the hypercube $[-\frac{R}{2}, \frac{R}{2}]^d$. An equivalent formulation is
\begin{equation} \label{Wintro}
\Wint(P) := \liminf_{R \ti} \frac{1}{R^d} \int_{[-R,R]^d \backslash \{0\}} g(v) (\rho_{2,P}(v)-1) \prod_{i=1}^d (R-|v_i|)dv,
\end{equation}
where $v = (v_1, \dots, v_d)$ and where we made again the abuse of notation $\rho_{2,P}(v) := \rho_{2,P}(0,v)$. The expression \eqref{Wintro} shows similarities with \eqref{def:Wper} in the logarithmic cases. 

Let us also define, in the logarithmic cases
\begin{equation} \label{def:Dlog}
\Dlog(\Pst) := \Clog \limsup_{R \ti} \left(\frac{1}{R^d} \iint_{\carr_R^2} (\rho_{2,P}(x,y) - 1) dx dy + 1\right) \log R,
\end{equation}
where $\Clog$ is a constant whose value is irrelevant for our concerns. 

Finally we introduce the free energy functional analogous to $\fbeta$ (defined in \eqref{def:fbeta}) 
$$\fbetai := \beta \Wint + \ERS[ \cdot | \Poisson],$$
or in the logarithmic cases
$$\fbetai := \beta (\Wint + \Dlog) + \ERS[ \cdot | \Poisson].$$

Let us recall the following definition: let $X$ be a topological space and $f,g: X \to \R$ two functions. We say that $g$ is the lower semi-continuous regularization of $f$ if for any $x \in X$ we have
\[
g(x) = \liminf_{y \to x} f(y).
\]

Our first main result is
\begin{theo} \label{theo:connec} The functionals $\Welec$ and $\Wint$ are related as follows.
\begin{itemize}
\item In the one-dimensional logarithmic case \eqref{wlog}, $\Welec$ is the lower semi-continuous regularization of $\Wint + \Dlog$, and for any $\beta \in (0, + \infty)$, $\fbeta$ is the lower semi-continuous regularization of $\fbetai$.
\item In the non-Coulomb cases \eqref{kernel} with $s > d-2$, $\Welec$ is the lower semi-continuous regularization of $\Wint$, and for any $\beta \in (0, + \infty)$, $\fbeta$ is the lower semi-continuous regularization of $\fbetai$.
\item In the two-dimensional logarithmic (Coulomb) case \eqref{wlog2d}, we have
\[
\Welec \leq \Wint + \Dlog
\]
\item In the higher dimensional Coulomb cases \eqref{kernel} with $s = d-2$, we have
\[
\Welec \leq \Wint.
\]
\end{itemize}
\end{theo}
A first interest of Theorem \ref{theo:connec} is that it provides a way of showing that a given random point process has finite energy. For example in the $d=3$ Coulomb case, the Poisson point process of intensity $1$ is easily seen to satisfy $\Wint(\Poisson) = 0$, hence $\Welec(\Poisson)$ is finite and nonpositive.

Let us emphasize that Theorem \ref{theo:connec} is less precise in the Coulomb cases than in the cases \eqref{wlog} and \eqref{kernel} with $s > d-2$, to which we will henceforth refer as the “non-Coulomb cases”.

In the following statement, by saying that two minimization problems are equivalent we mean that both functionals have exactly the same infima. If $g$ is the lower semi-continuous regularization of $f$ on $X$, then the minimization problems associated to $f$ and $g$ are equivalent, thus
\begin{coro} \label{coro:fbetai} We deduce from Theorem \ref{theo:connec} that in the non-Coulomb cases
\begin{enumerate}
\item Minimizing $\Welec$ on $\probas_{s,1}(\config)$ is equivalent to minimizing $\Wint$ (or $\Wint + \Dlog$) on $\probas_{s,1}(\config)$.
\item For any $\beta \in (0, + \infty)$, minimizing $\fbeta$ on $\probas_{s,1}(\config)$ is equivalent to minimizing $\fbetai$ on $\probas_{s,1}(\config)$.
\end{enumerate}
\end{coro}

\paragraph{Applications.}
We use the explicit expression of $\Wint$ and its link with $\Welec$ to perform simple computations, which allow us to prove the following two results concerning the minimization of $\fbeta$ in the limit $\beta \t0$ and $\beta \ti$. 
\subparagraph{High-temperature limit.} As can be expected, when $\beta \t0$ any minimizer of $\fbeta$ gets close to the minimizer of the entropy term.
\begin{theo} \label{theo:poisson} For all cases \eqref{wlog}, \eqref{wlog2d}, \eqref{kernel}, the minimizers of $\fbeta$ converge as $\beta \t0$ to the law of the Poisson point process $\Poisson$. Moreover this convergence holds in entropy sense i.e.
\begin{equation} \label{toPoisson}
\lim_{\beta \to 0} \sup_{\fbeta(P_{\beta}) =  \min \fbeta} \ERS[P_{\beta}|\Poisson] = 0.
\end{equation}
\end{theo}
In the special case of one-dimensional log-gases, as proven in \cite{LebSer} a minimizer of $\fbeta$ is the $\sineb$ process of Valko-Virag \cite{vv}. Hence our method yields another proof for a recent result of Allez and Dumaz \cite{AllezDumaz}:
\begin{coro} \label{coro:sineb} As $\beta \to 0$ the $\sineb$ point process converges weakly in the space of Radon measure (endowed with the topology of vague convergence) to the law of a Poisson point process on $\R$.
\end{coro} 
\subparagraph{Low-temperature limit.}
In dimension $1$ we may also characterize the limit $\beta \ti$ (the low temperature limit) of the minimizers of $\fbeta$. We let $P_{\Z}$ be the stationary random point process associated to the lattice $\Z$ 
\begin{equation}
\label{def:PZ} P_{\Z} := \int_{0}^1 \delta_{x + \Z}\ dx,
\end{equation}
which can also be seen as the law of the point configuration $u + \Z$ where $u$ is a uniform random variable in $[0,1]$ and where we let $x + \Z$ denote the point configuration $\{x + k, k \in \Z\}$.

\begin{theo}[Crystallization for $d=1$] \label{theo:minim1d} For $d=1$ and in both cases \eqref{wlog} or \eqref{kernel}, the random point process $P_{\Z}$ is the unique minimizer of $\Welec$ on $\probas_{s,1}(\config)$. Moreover if $\{P_{\beta}\}_{\beta}$ is a family of minimizers of $\fbeta$, we have
\begin{equation} \label{toZ}
\lim_{\beta \t0} P_{\beta} = P_{\Z}.
\end{equation}
\end{theo}
Theorem \ref{theo:minim1d} is a \textit{crystallization} result, proving the convergence to the one-dimensional crystal as $\beta \ti$. A similar result was proven in the one-dimensional logarithmic case in \cite{Leble1d} by using the explicit expression available in the periodic setting together with an approximation by periodic point processes. The method here is similar in spirit but follows a simpler approach which works for Riesz cases as well.

\subsection{Open questions}
Let us now briefly mention some questions that are raised by, or related to our present study.

\paragraph{Minimizers of $\Wint$.} Could one determine the infimum of $\Wint$ (thus of $\Welec$ in the non-Coulomb cases) thanks to its explicit form? A necessary condition on the realizability of $\rho_2$ as a two-point correlation function of some stationary random point process of intensity one is that $T_2 := \rho_2 -1$ satisfies 
\[
T_2 \geq -1, \quad \widehat{T_2} \geq -1,
\]
(see \cite[Section 2]{Kuna}) where $\widehat{T_2}$ denotes the Fourier transform of $T_2$ (in a sense that should, in general, be precised). Thus we may start by asking whether the linear optimization problem of minimizing
\[
\liminf_{R \ti} \int_{[-R,R]^d} T_2(v) g(v) \prod_{i=1}^d \left(1 - \frac{|v_i|}{R}\right) dv
\]
can be solved on the convex set $\{T_2 \geq -1, \widehat{T_2} \geq -1\}$. It is unclear to us whether the symmetry of the constraints on $T_2$ and $\widehat{T_2}$ might be of any use (however let us observe that the two-point correlation function of $P_{\Z}$ is its own Fourier transform, and that the expected minimizer in $d=2$, namely the stationary random point process associated to the triangular lattice, exhibits a self-duality of the same kind).

\paragraph{Decorrelating random point processes.} We may also investigate the problem of minimizing $\Wint$ over particular sub-classes of $\probas_{s,1}(\config)$. We have already mentioned the work of \cite{BorSer} where this minimization is considered over determinantal point processes in dimension $d= 1,2$. Another interesting aspect is that of “decorrelating” random point processes. Let us say that $P$ is a \textit{decorrelating} random point process (or that $P$ \textit{decorrelates}) when $T_2(v) := \rho_2(v) -1$ tends to $0$ as $|v| \ti$, with some speed criterion to be fixed, or when $T_2$ lies in some reasonable class (e.g. $L^p$-spaces). It is unclear to us whether there is a minimizing sequence for $\Wint$ made of random point processes which decorrelate. A negative answer would hint at a soft kind of phase transition as $\beta$ varies. Indeed as $\beta \t0$ the minimizers of $\fbeta$ converge to the law of a Poisson point process $\Poisson$ (according to Theorem \ref{theo:poisson}), which is a typical decorrelating random point process, and as $\beta \ti$ they would rather leave this class. 

\subparagraph{The $1d$ Log-gas case.}
Such a transition may be formally observed in the one-dimensional logarithmic case. Let us recall that for $d=1, s=0$ there exists a “concrete” family of minimizers for $\fbeta$ stemming from Random Matrix Theory (RMT), namely the $\sineb$ processes. They arise as the $N \ti$ limit of microscopic point processes observed in the so-called “Gaussian $\beta$-ensembles”, which are a generalization for any $\beta$ of the classical Gaussian ensembles of RMT (see \cite{de} and \cite{vv}). It is known (see \cite{na14}) that $\sineb$ is also the law of the limiting microscopic point process for the “Circular $\beta$-ensemble”, which is another RMT model. The $N \ti$ limit for two-point correlation function of the Circular $\beta$-ensemble has been derived in \cite{Forrester1} for even values of $\beta$. Equation \cite[(11a)]{Forrester1} should thus describe the large-$x$ asymptotics of the two-point correlation function of $\sineb$, i.e. formally we have
\[
T_{2, \sineb}(x) := \rho_{2, \sineb}(x) - 1 \approx \sum_{k=1}^{\beta/2} f_k(x) x^{-4k^2/\beta}, 
\]
where $f_k(x)$ is a slowly oscillating function. The leading term of $T_{2, \sineb}(x)$ as $x \ti$ is of order $x^{-4/\beta}$, we would thus expect $T_{2, \sineb}$ to leave the class $L^{p}$ as soon as $\beta \geq 4p$.

\subparagraph{Decorrelating random point process of minimal energy.}
A negative answer would also raise the question of finding the minimizer (or a minimizing sequence) for $\Wint$ among decorrelating random point processes. It seems to us that a good candidate for a lower bound on the energy is given by the hypothetical “hardcore Poisson point process” $\Poisson_{\rm{hc}}$ whose two-point correlation function would be $\rho_{2,\rm{hc}} = 1 - \mathbf{1}_{B}$ where $B$ is the ball of center $0$ and unit volume in $\R^d$. In dimension $d=1$ it is not hard to construct a sequence of random point processes whose energies converge to the associated energy $\int_{B} \log|v|$. In arbitrary dimension, it easy to see that any sub-Poissonian random point process (i.e. a random point process such that $\rho_2 \leq 1$, e.g. any determinantal point process) has a larger energy than $\Poisson_{\rm{hc}}$.

\paragraph{Plan of the paper and ideas of proof.}
In Section \ref{sec:defnot} we give some general definitions and notation. 

In Section \ref{sec:defW} we recall the definition of the \textit{renormalized energy} in the sense of \cite{SS2d}, \cite{RougSer}, \cite{PetSer}, then we introduce the alternative object $\Wint$ defined on the space of random point processes, and in the stationary case we give a simple expression of $\Wint(P)$ in terms of the two-point correlation function of $P$. 

In Section \ref{sec:prelim} we give some preliminary results. In particular we observe that while $\Welec$ is by definition computed in terms of the energy of \textit{global} electric fields defined on the whole space $\R^{d+k}$, the object $\Wint$ is rather a limit as $R \ti$ of the energy of \textit{local} electric fields defined on hypercubes of sidelength $R$.

Section \ref{sec:connection} is devoted to the proof of Theorem \ref{theo:connec}. The proof goes in two step: first we show that $\Welec \leq \Wint$ (or $\Wint + \Dlog$ in the logarithmic cases) on the space $\probas_{s,1}(\config)$, then conversely (in the non-Coulomb cases) for any $P \in \probas_{s,1}(\config)$ we prove the existence of a “recovery sequence” $\{P_N\}_{N}$ converging to $P$ and such that $\lim_{N \ti} \Wint(P_N) \leq \Welec(P)$. In both steps the key element is the \textit{screening lemma} of \cite{PetSer} (following \cite{SS2d}, \cite{SS1d}, \cite{RougSer}) which, heuristically speaking, allows us here to construct a global electric field from a local one, and \textit{vice versa}.

In Section \ref{sec:HT} we prove Theorem \ref{theo:poisson} about the convergence to the law of the Poisson point process $\Poisson$ of minimizers of $\fbeta$ as $\beta \t0$. If $\Welec(\Poisson)$ is finite we may see directly that $\ERS[P_{\beta}|\Poisson]$ must go to zero as $\beta \t0$, and then the \textit{specific} Pinsker inequality implies that $P_{\beta} \to \Poisson$ as $\beta \t0$. However in some cases the finiteness of $\Welec(\Poisson)$ is false (e.g. $d=1, s=0$, see \cite{LebSer}) or yet unknown ($d=2, s=0$). We use the fact that $\Welec \leq \Wint + \Dlog$  to construct a sequence of random point processes converging to $\Poisson$ in entropy sense and whose renormalized energy is finite. 

In Section \ref{sec:LT} we restrict ourselves to the one-dimensional cases \eqref{wlog} or \eqref{kernel} and we prove the crystallization result of Theorem \ref{theo:minim1d}. We start by  using a convexity argument to show that $P_{\Z}$ is the only minimizer of $\Wint$ over $\probas_{s,1}(\config)$. More precisely we obtain a quantitative bound below on $\Wint(P) - \Wint(P_{\Z})$ in terms of the two-point correlation function of $P$. This translates into a bound below for $\Welec(P) - \Welec(P_{\Z})$ which implies that $P_{\Z}$ is also the only minimizer of $\Welec$. Moreover as $\beta \ti$ we show that $\Welec(P_{\beta})$ must go to $\Welec(P_{\Z})$, which in turn implies that as $\beta \ti$ the two-point correlation function of $P_{\beta}$ converges to $\rho_{2, P_{\Z}}$ in the distributional sense. Thanks to the “rigidity” of the lattice it is not hard to deduce that in fact $P_{\beta}$ converges to $P_{\Z}$ as $\beta \ti$.

\paragraph{Aknowledgements.}
The author would like to thank his PhD supervisor, Sylvia Serfaty, for helpful discussions and many useful comments on this paper. 

\section{Definitions and notation} \label{sec:defnot}
\subsection{Generalities} 
If $(X,d_X)$ is a metric space we let $\probas(X)$ be the space of Borel probability measures on $X$ and we endow $\probas(X)$ with the distance
\begin{equation} \label{def:Dudley}
d_{\probas(X)}(P_1, P_2) = \sup \left\lbrace \int F (dP_1 - dP_2) |\  F \in \Lip(X) \right\rbrace,
\end{equation}
where $\Lip(X)$ denotes the set of functions $F : X \rightarrow \R$ that are $1$-Lipschitz with respect to $d_X$ and such that $\|F\|_{\infty} \leq 1$. It is well-known that the distance $d_{\probas(X)}$ metrizes the topology of weak convergence on $\probas(X)$. If $P \in \probas(X)$ is a probability measure and $f: X \rightarrow \Rd$ a measurable function, we denote by $\Esp_P \left[ f \right]$ the expectation of $f$ under $P$.

The ambient Euclidean space of dimension $d$ is denoted by $\Rd$. We will often need to work in $\R^{d+k}$ where $k=0$ (in the case of Coulomb interactions) or~$1$ (see Section \ref{sec:extension}).

For any $R > 0$ we denote by $\carr_R$ the hypercube $[-R/2, R/2]^d \subset \R^d$ and by $\hcarr_R$ the hypercube $[-R/2, R/2]^{d+k} \subset \R^{d+k}$. Similarly we let $B_R$ be the ball of center $0$ and radius $R$ in $\Rd$ and $\hB_R$ be the ball of center $0$ and radius $R$ in $\R^{d+k}$.

If $A$ is a set, we let $\triangle := \{(x,x), x \in A\} \subset A \times A$ be the diagonal of $A$.

\subsection{Point configurations and random point processes}
\paragraph{Point configurations.}
If $A$ is a Borel set of $\Rd$ we denote by $\config(A)$ the set of locally finite point configurations in $A$ or equivalently the set of non-negative, purely atomic Radon measures on $A$ giving an integer mass to singletons (see \cite{dvj}). We will often write  $\mathcal C$ for $\sum_{p \in \mathcal C} \delta_p$. We endow the set $\config := \config(\Rd)$ (and the sets $\config(A)$ for $A$ Borel) with the topology induced by the topology of weak convergence of Radon measure (also known as vague convergence or convergence against compactly supported continuous functions), these topologies are metrizable and we fix an arbitrary compatible distance.

The additive group $\Rd$ acts on $\config$ by translations $\{\theta_t\}_{t \in \Rd}$: if $\mathcal{C} = \{x_i, i \in I\} \in \config$ we let 
\begin{equation}\label{actiontrans}
\theta_t \cdot \mathcal{C} := \{x_i - t, i \in I\}.
\end{equation}

We denote by $\Nn_R : \config \mapsto \N$ the number of points of a configuration in the hypercube $\carr_R$, and by $\D_R$ the discrepancy $\D_R = \Nn_R - R^d$.

\paragraph{Random point processes.}
A random point process is a probability measure on $\config$. We denote by $\probas_s(\config)$ the set of translation-invariant (or stationary) random point processes. We endow $\probas(\config)$ with the topology of weak convergence of probability measures on $\config$. A compatible distance on $\probas_s(\config)$ is defined in \eqref{def:Dudley}.
\begin{remark} \label{rem:memetopologie} Another natural topology on $\probas(\config)$ is the “convergence of the finite distributions” \cite[Section 11.1]{dvj2}, also called the “convergence with respect to vague topology for the counting measure of the point process”. The two notions of convergence coincide as stated in \cite[Theorem 11.1.VII]{dvj2}.
\end{remark}

We will use several times the operation of averaging a random point process over translations in some measurable subset. If $P$ is a random point process in $\Rd$ and $K$ a measurable subset of $\Rd$ with finite, non-zero Lebesgue measure, we define the average $\Pav_K$ of $P$ over translations in $K$ as the law of the random variable $u_K + \C$ where $u_K$ is uniformly distributed according to the normalized Lebesgue measure on $K$, and $\C$ has law $P$. The sum of a vector and a point configuration is defined as the point configuration
$$ x + \C := \{ x + p, p \in \C \}.$$

\paragraph{Hyperuniformity.} Following \cite{torquato} we say that a random point process $P$ is \textit{hyperuniform} if we have
\begin{equation} \label{def:hyperun}
\Esp_{P}[\D^2_R] = O(R^{d-1}).
\end{equation}

For a stationary one-dimensional random point process $P$, hyperuniformity is easily seen to be equivalent to the following property: for some $r > 0$, there are $P$-a.s. between $k-r$ and $k+r$ points in any interval of length $k$.

\subsection{Correlation functions} \label{sec:correlfun}
Let $P \in \probas(\config)$ be a random point process. For any $n \geq 1$ the $n$-point correlation function $\rho_{n,P}$ is the linear form on (a subspace of) the linear space of bounded measurable functions $\varphi_{n} : \R^n \longrightarrow \R$ with compact support defined by (we abbreviate “p.d.” for “pairwise distinct”)
\begin{equation} \label{def:rhon}
\rho_{n, P}(\varphi_{n}) = \Esp_{P} \sum_{x_1, \dots, x_n \in \C \text{ p.d.}} \varphi_{n}(x_1, \dots, x_n).
\end{equation}
If the $n$-point correlation function exists as a distribution and can be identified with a measurable function, we will write $\int \rho_{n,P} \varphi_n$ instead of $\rho_{n,P} (\varphi_n)$. Heuristically speaking, the one-point correlation function $\rho_1$ (also called the intensity of the random point process) gives the density of the process at each point, while the two-point correlation function $\rho_2(x,y)$ gives the probability of having a point both at $x$ and $y$. In this paper we will work with stationary random point processes such that $\rho_1 \equiv 1$ and we denote by $\probas_{s,1}(\config)$ this set.
 
\subsection{Dimension extension} \label{sec:extension} We recall some elements from \cite{PetSer} to which we refer for more details. Outside of the Coulomb cases the Riesz kernel $g$ is not the convolution kernel of a local operator, but rather of a fractional Laplacian. It can be transformed into a local but inhomogeneous operator of the form $\div (\yg \nab \cdot)$ by adding one space variable $y \in \R$ to the space $\R^d$. In what follows, $k$ will denote the dimension extension. We take $k=0$ in all the Coulomb cases, i.e. $s=d-2$ and $d \ge 3$ or \eqref{wlog2d} and in all other cases we take $k=1$. We use an auxiliary parameter $\gamma$ defined by
\begin{equation}
\label{def:gamma} \gamma := s - d + 2 - k
\end{equation}
where the convention is to take $s=0$ in the logarithmic cases. In particular we have $\gamma = 0$ in the logarithmic cases \eqref{wlog} (where $k=1$) and \eqref{wlog2d} (where $k=0$).

Points in the space $\R^d$ will be denoted by $x$, and points in the extended space $\R^{d+k}$ by $X = (x,y)$, $x\in \R^d$, $y\in \R^k$. The interaction kernel $g$ is naturally extended to $\R^{d+k}$. We will often identify $\R^d \times \{0\}$ and $\R^d$. The measure $\delta_{\R^d}$ is the Radon measure on $\R^{d+k}$ which corresponds to the Lebesgue measure on the hypersurface $\R^d \subset \R^{d+k}$.  

Finally we let $\cds$ be the constant depending on $d,s$, such that $ - \div( \yg \nabla g) = \cds \delta_0$ 
in $\R^{d+k}$ (the values of $\cds$ are given in \cite{PetSer}).

\subsection{Electric fields and random electric fields}
Let $p< \frac{d+k}{s+1}$ be fixed. We think of the space $\Lploc(\R^{d+k}, \R^{d+k})$ as the space of electric fields generated by the charged particles together with a certain uniformly charged background and we endow this space with the weak $L^p$ topology. 

\paragraph{Local electric fields.} 
If $\C$ is a finite point configuration and $R > 0$ we let
\begin{equation}
\label{def:Eloc}
\Hloc(\C) := \cds g * (\C - \mathbf{1}_{C_R} \delta_{\Rd} ), \quad \Eloc(\C) := \nabla \Hloc = \cds \nabla g * (\C - \mathbf{1}_{C_R} \delta_{\Rd}).
\end{equation}
where $g *$ denotes the convolution (computed in $\R^{d+k}$) with the interaction kernel $g$. It implies that
\begin{equation} \label{relationEloc}
- \div( \yg \Eloc) = \cds \nabla g * (\C - \mathbf{1}_{C_R} \delta_{\Rd}).
\end{equation}

The scalar field $\Hloc$ physically corresponds to the electrostatic potential generated by the point charges of $\C$ together with a  background of density $\mu$. The vector field $\Eloc$ can be thought of as the associated electrostatic field. It is easy to see that $\Eloc$ fails to be in $L^2_{\rm{loc}}$ because it blow ups like $|x|^{-(s+1)}$ near each point of $\C$, however $\Eloc$ is in $\Lploc(\R^{d+k}, \R^{d+k})$.

\paragraph{Electric fields.} We now introduce a special class of vector fields that correspond to electric fields generated by a system made of an infinite point configuration $\C$ and a negatively charged background in all $\Rd$. We let $\A$ be the class of “electric vector fields” i.e. the set of vector fields $E$ belonging to $\Lploc(\R^{d+k}, \R^{d+k})$ that satisfy
\begin{equation}\label{Eadmiss}
- \div (\yg E) = \cds \left( \C - \delta_{\Rd} \right)  \text{ in } \R^{d+k}
\end{equation}
where $\C \in \config(\Rd)$ is a point configuration. We say that \textit{$E$ is compatible with $\C$} if \eqref{Eadmiss} holds. If $E$ is in $\A$ we let 
$$\mathrm{Conf}(E) := \frac{-1}{\cds} \div (\yg E) + \delta_{\Rd}$$
 be the underlying point configuration $\C$, in other words $E$ is compatible with $\C$ if and only if $\conf(E) = \C$. 

\paragraph{Truncation procedure.}
The renormalization procedure of \cite{RougSer}, \cite{PetSer} (inspired by the original work of \cite{bbh}) uses a truncation of the singularities which we now recall. We define the truncated Riesz (or Coulomb, or logarithmic) kernel as follows:
for  $1>\eta>0$ and  $X\in \R^{d+k}$, let 
\begin{equation}
\label{def:feta} f_\eta(X)= \left(\g(X)- \g(\eta)\right)_+
\end{equation}
with a slight abuse of notation: since $g$ is a radial function we write $g(\eta)$ for the value of $g$ at any point on a sphere of radius $\eta$.

If $\Eloc$ is a local field as in \eqref{def:Eloc} we let 
\begin{equation}
\label{def:Eloceta} \Eloc_{\eta}(X) := \Eloc(X) - \sum_{p \in \C} \nabla f_\eta (X-p).
\end{equation}

Similarly if $E$ is an electric field as in \eqref{Eadmiss} we let
\begin{equation}
\label{def:Eeta} E_{\eta}(X) := E(X) - \sum_{p \in \C} \nabla f_\eta (X -p).
\end{equation}

\paragraph{Random electric process.}
A probability measure on $\probas(\Lploc(\R^{d+k}, \R^{d+k}))$ concentrated on $\A$ is called a random electric process. We say that $\Pelec$ is stationary when it is invariant under the (push-forward by) translations $E \mapsto E(\cdot - x)$ for any $x \in \R^d$. 

\subsection{Specific relative entropy} \label{sec:ERS}
Let $\Pst$ be a stationary random point process on $\Rd$. The relative specific entropy $\ERS[\Pst|\Poisson]$ of $\Pst$ with respect to $\Poisson$, the law of the Poisson point process of uniform intensity $1$, is defined by
\begin{equation} \label{def:ERS}
\ERS[\Pst|\Poisson] := \lim_{R \ti} \f{1}{|\carr_R|} \Ent\left(\Pst_{|\carr_R} | \Poisson_{|\carr_R} \right),
\end{equation}
where $P_{|\carr_R}$ denotes the random point process induced in $\carr_R$, and $\Ent( \cdot | \cdot)$ denotes the usual relative entropy (or Kullbak-Leibler divergence) of two probability measures defined on the same probability space. We take the appropriate sign convention for the entropy so that it is non-negative: if $\mu,\nu$ are two probability measures defined on the same space  we let $\Ent \left(\mu | \nu \right) := \int \log \frac{d\mu}{d\nu} d\mu$ if $\mu$ is absolutely continuous with respect to $\nu$ and $+ \infty$ otherwise. We have in fact by super-additivity
\begin{equation} \label{superadditive}
\ERS[\Pst|\Poisson] = \sup_{R \geq 1} \f{1}{|\carr_R|} \Ent\left(\Pst_{|\carr_R} | \Poisson_{|\carr_R} \right).
\end{equation} 
The functional $\Pst \mapsto \ERS[\Pst|\Poisson]$ is affine lower semi-continous on $\probas_{s,1}(\config)$ and its sub-level sets are compact. We refer to \cite[Chap. 6]{seppalainen} for a proof of these statements.

\section{Definitions for the energy of a random point process}  \label{sec:defW}
In this section we recall the derivation of a renormalized energy for random point processes from \eqref{def:WN} then we introduce our alternative definition of a (logarithmic, Coulomb or Riesz) energy for random point processes. 

\subsection{The electric approach} \label{sec:WSS}
\paragraph{Renormalized energy of an electric field.}
We now recall the computation of the \textit{renormalized} energy $\mc{W}(E)$ following \cite{PetSer} (see also \cite{serfatyZur} and the references therein). For any $E \in \A$ we define 
\begin{equation} \label{def:Weta}
\mc{W}_{\eta}(E) := \limsup_{R\to \infty} \left( \frac{1}{\cds} \frac{1}{R^d}\int_{\carr_R\times \R^k} \yg |E_\eta|^2 - \g(\eta)\right),
\end{equation}
where $E_{\eta}$ is the truncated field as in \eqref{def:Eeta}, and we let
\begin{equation} \label{def:tW}
\mc{W}(E) := \lim_{\eta\to 0} \mc{W}_{\eta}(E).
\end{equation}
For convenience we have chosen a definition of $\mc{W}$ which differs from a multiplicative constant $\cds$ from the one in \cite{PetSer}.

\paragraph{The electric definition.} 
We then define 
\begin{equation} \label{def:WdeE}
\W_{\eta}(\C) = \inf \mc{W}_{\eta}(E), \quad \W(\C) = \inf \mc{W}(E),
\end{equation}
where both infimum are among electric fields $E$ compatible with $\C$. Similarly if $P$ is a random point process we let
\begin{equation} \label{def:Welec}
\Welec_{\eta}(P) = \Esp_{P} \left[ \W_{\eta} \right], \quad \Welec(P) = \Esp_{P} \left[ \W \right].
\end{equation}

 The following lemma was proven in \cite{LebSer}:
\begin{lem} \label{lem:concordWelec}
Let $\Pst$ be a stationary random point process such that $\Welec(P)$ is finite. Then there exists a stationary random electric process $\Pelec$ such that the push-forward of $\Pelec$ by $\conf$ is equal to $\Pst$ and which satisfies
\begin{equation} \label{concordWelec}
\Esp_{\Pelec} [ \mc{W} ] = \Welec(P).
\end{equation}
\end{lem}

We also have the following lower semi-continuity result for the electric energy. Let us emphasize that lower-semi continuity only holds at the level of \textit{stationary} random point processes, and not for point configurations or arbitrary random point processes.
\begin{lem} \label{lem:lsciW}
The maps $P \mapsto \Welec_{\eta}(P)$ and $P \mapsto \Welec(P)$ are lower semi-continuous on $\probas_{s,1}(\config)$.
\end{lem}
\begin{proof}
See e.g. \cite[Lemma 4.1.]{PetSer}
\end{proof}

This provides a definition for the energy of a random infinite point configuration with a uniform negative background. However the computation of $\Welec(P)$ or of $\W(\C)$ (or even the search for an upper bound on these quantities) appears involved in general because it amounts to finding compatible electric fields for infinite point configurations. Let us mention that in the case of a periodic configuration (hence also for the stationary random point process associated to it) exact formulas are known in the cases \eqref{wlog2d} (see \cite{SS2d}), \eqref{wlog} (see \cite{SS1d}), for the higher-dimensional Coulomb case \eqref{kernel} with $s = d-2$ (see \cite{RougSer}) and for Riesz gases (see \cite{PetSer}).

\subsection{The intrinsic approach} \label{sec:Wint}
In this section we define an energy functional $\Wint$ on the space of random point processes, using only the nature of the pairwise interaction. 

If $A, B$ are two (measurable) subsets of $\Rd$ we define $\Inte[A,B]$ as the interaction energy between $A$ and $B$
\begin{equation} \label{def:inte}
\Inte[A, B](\C) := \iint_{(A \times B) \backslash \triangle} g(x-y) (d\C(x) - dx) \otimes (d\C(y) - dy).
\end{equation}

In view of \eqref{def:WN}, for any $R > 0$ we define $\Hint_{R} : \config(\carr_R) \rightarrow \R$ as the interaction of $\carr_R$ with itself (the diagonal being excluded) i.e. $\Hint_{R} := \Inte[C_R, C_R]$ or in other terms
\begin{equation}
\label{def:Hint}
\Hint_R(\C) := \iint_{\carr_R^2 \backslash \triangle} g(x-y) (d\C - dx)\otimes(d\C - dy).
\end{equation}
Given a random point process $\Pst$, a natural way of defining the energy (per unit volume) of $\Pst$ is the following:
\begin{defi} \label{def:Wint}
Let $P$ be a random point process of intensity $1$. We define its intrinsic energy $\Wint$ by
\begin{equation} \label{eqdef:Wint}
\Wint(P) := \liminf_{R \ti} \frac{1}{R^d} \Esp_P \left[\Hint_{R}(\C) \right]. 
\end{equation}
\end{defi}

\paragraph{Expression with correlation functions.}
The energy defined by \eqref{eqdef:Wint} can be re-written with the help of one- and two-point correlation functions of $P$. In order for the expression to make sense, we restrict ourselves to random point processes whose two-point correlation function exists as a Radon measure in $\Rd \times \Rd$. We will then abuse notation and consider $\rho_2$ as a function instead of a measure, writing $\rho_2(x,y)$ instead of $d\rho_2$.

\begin{lem} For any random point process $P$ of intensity $1$ such that $\rho_2$ exists as a Radon measure, the following identity holds
\begin{equation} \label{rewriWint}
\Wint(P) = \liminf_{R \ti} \f{1}{R^d} \iint_{\carr_R^2 \backslash \triangle} g(x-y) (\rho_2(x,y)-1) dxdy.
\end{equation}
\end{lem}
\begin{proof} First we may re-write $\Hint_{R}(\C)$ as : 
\begin{multline*}
\Hint_{R}(\C) = \iint_{\carr_R^2 \backslash \triangle} g(x-y) (d\C - dx)\otimes(d\C - dy)  \\ = \iint_{\carr_R^2 \backslash \triangle} g(x-y) (d\C\otimes d\C) + \iint_{\carr_R^2 \backslash \triangle} g(x-y) (dx \otimes dy) - 2\iint_{\carr_R^2 \backslash \triangle} g(x-y) d\C \otimes dy
\end{multline*}
and by definition of the correlation functions (see Section \ref{sec:correlfun}) we get
\begin{equation*}
\Esp_{P}\left[ \Hint_R \right] = \iint_{\carr_R^2 \backslash \triangle} g(x-y) \rho_2(x,y)dxdy + \iint_{\carr_R^2 \backslash \triangle} g(x-y) dx dy - 2 \iint_{\carr_R^2 \backslash \triangle} g(x-y) \rho_1(x)dx dy.
\end{equation*}
By assumption $P$ has intensity $1$ i.e. $\rho_1 \equiv 1$ and we are left with
\begin{equation*}
\frac{1}{R^d} \Esp_{P}\left[ \Hint_R \right] = \frac{1}{R^d} \iint_{\carr_R^2 \backslash \triangle}  g(x-y) (\rho_2(x,y)-1) dx dy,
\end{equation*}
which yields \eqref{rewriWint}.
\end{proof}

\paragraph{The stationary case.}
If the random point process $P$ is stationary we may derive a somewhat simpler expression for $\Wint(P)$. In what follows $\rho_2(v)$ stands for $\rho_2(0,v)$. The change of variables $(u,v)= (x+y,x-y)$ gives 
\begin{multline*}
\iint_{([-R/2,R/2]^d)^2 \backslash \triangle} g(x-y) (\rho_2(x-y)-1)dx dy = \frac{1}{2^d} \int_{v\in [-R,R]^d \backslash \{0\}} \int_{u \in S_R(v)} g(v) (\rho_2(v)-1)du dv \\ 
= \int_{v\in [-R,R]^d \backslash \{0\}} |\SR(v)| g(v) (\rho_2(v)-1) dv,
\end{multline*}
where $\SR(v)$ denotes the set 
$$\SR(v) :=\{x+y: x, y \in \carr_R,\, x-y=v\}.$$ 
The Lebesgue measure $|\SR(v)|$ of $\SR(v)$ is easily computed for $v = (v_1, \dots, v_d)$ 
\begin{equation} \label{volSR}
\left|\SR(v_1, \dots, v_d)\right| = 2^d \times \prod_{i=1}^d (R-|v_i|).
\end{equation}
Indeed we have $\SR(v) =\{2x-v: x, y \in \carr_R,\, x-y=v\}$ which implies that
$$|\SR(v)| = 2 \times \left|\{x \in\carr_R, x-v \in \carr_R\}\right|,$$
moreover $\SR(v)$ tensorizes i.e. if $v=(v_1, \dots, v_d) \in [-R, R]^d$ we get
\begin{equation}
|\SR(v_1, \dots, v_d)| = 2^d \prod_{i=1}^d \left|\{x_i \in [-R/2,R/2], x_i-v_i \in [-R/2,R/2]\}\right|
\end{equation} 
which leads to \eqref{volSR}. Finally we obtain the following expression for any $P$ in $\probas_{s,1}(\config)$ (such that $\rho_2$ is a Radon measure)
\begin{equation} \label{def:Wint2}
\boxed{\Wint(P) = \liminf_{R \ti} \frac{1}{R^d} \int_{[-R,R]^d \backslash \{0\}} g(v) (\rho_2(v)-1) \prod_{i=1}^d (R-|v_i|)dv.}
\end{equation}

\section{Preliminary results on the energy} 
\label{sec:prelim}
\subsection{Local field and local interaction} \label{sec:boundary}
Let $R > 0$ and let $\C$ be a point configuration in $\config(\carr_R)$. The local electric potential (resp. field) $\Hloc$ (resp. $\Eloc$) are defined in \eqref{def:Eloc}. For $\eta \in (0,1)$ we let also $\Eloc_{\eta}$ be as in \eqref{def:Eloceta}. Let us recall that $\hC_R$ denotes the hypercube $[-R/2, R/2]^{d+k}$. 
\begin{lem} \label{lem:locflocint}
The following inequalities hold
\begin{enumerate}
\item In the cases \eqref{kernel} 
\begin{equation} \label{compkernel}
\frac{1}{\cds} \int_{\R^{d+k}} \yg |\Eloc_{\eta}|^2 - \Nn_R g(\eta) \leq \Hint_R[\C] + \Nn_R o_{\eta}(1).
\end{equation}
\item In the logarithmic cases 
\begin{multline} \label{complog}
\frac{1}{\cds} \int_{\hC_R} \yg |\Eloc_{\eta}|^2 - \Nn_R g(\eta) \leq \Hint_R[\C] + \Clog \D^2_R \log R  + \Nn_R o_{\eta}(1) \\ + O(\Nn^2_RR^{-5}) + o_R(1).
\end{multline}
\end{enumerate}

Let $P$ be a stationary random point process of intensity $1$ such that $\Esp_{P}[\D^2_R] = o(R^{2d})$. The following inequalities hold
\begin{enumerate}
\item In the cases \eqref{kernel} 
\begin{equation}
\label{relatkern}
\Esp_{P} \left[ \frac{1}{\cds} \int_{\R^{d+k}} \yg |\Eloc_{\eta}|^2 - \Nn_R g(\eta)\right] \leq \Esp_{P}[\Hint_{R}]  + R^d o_{\eta}(1).
\end{equation}
\item In the logarithmic cases
\begin{multline}
\label{relatlog}
\Esp_{P} \left[\frac{1}{\cds} \int_{\hC_R} \yg |\Eloc_{\eta}|^2 - \Nn_R g(\eta)\right] \leq \Esp_{P}[\Hint_{R}] + \Clog \Esp_{P}[\D_R^2]\log R + R^d o_{\eta}(1) \\ + o_R(1).
\end{multline}
\end{enumerate}

The terms $o_{\eta}(1), o_R(1)$ depends only on $d,s$. 
\end{lem}

\begin{proof} The starting point is the following identity which holds for any $S > R$
\begin{equation} \label{energienonneutre}
\Hint_R[\C] = \lim_{\eta \t0} \left(\frac{1}{c_{d,s}} \int_{\hC_S} \yg |\Eloc_{\eta}|^2 - \Nn_R g(\eta) \right) \\
+ \frac{1}{\cds} \int_{\partial \hB_S} \yg \Hloc \Eloc \cdot \vec{n},
\end{equation}
where $\vec{n}$ denotes the unit normal vector. It results from two different operations. Let us first recall that
\[
\Hint_R[\C] := \iint_{\carr_R^2 \backslash \triangle} g(x-y) (d\C - dx)\otimes(d\C - dy)
\]
Since $\Eloc = \nabla \Hloc$ satisfies \eqref{relationEloc} we may formally write
\[
\Hint_R[\C] \approx - \frac{1}{\cds} \int_{\R^{d+k}} \Hloc(t) \div(\yg \nabla \Hloc)(t).
\]
Of course $\div(\yg \nabla \Hloc)(t)$ is in fact supported on $\carr_R$. For any $S > R$ an integration by parts yields
\[
\Hint_R[\C] \approx \frac{1}{\cds} \int_{\hC_S} \yg |\Eloc|^2 +  \frac{1}{\cds} \int_{\partial \hC_S} \yg \Hloc \Eloc \cdot \vec{n}.
\]
Since $\Eloc$ is not in $L^2$ the previous computation does not make sense, however it can be made rigorous in a \textit{renormalized} fashion by truncating the interaction close to the charges at scale $\eta > 0$ and substracting a diverging term as $\eta \t0$ (here $g(\eta)$) for each charge. We refer to \cite{PetSer} for more details.

We now turn to the boundary term in \eqref{energienonneutre}. A mean value argument applied to $g$ and $\nabla g$ shows that for $S \geq 2R$ we have
\[
\left|\Hloc(X) - \D_R g(X)\right| \leq C (\Nn_R + R^d) R S^{-s-1},
\]
\[
|\Eloc(X) - \D_R \nabla g(X)| \leq C (\Nn_R + R^d) R S^{-s-2},
\]
uniformly for $X \in \partial \hC_S$. Indeed it is easy to see that the first derivative of $t \mapsto g(X-t)$ is bounded by $CS^{-s-1}$ and that its second derivative is bounded by $CS^{-s-2}$ for any $t$ in $C_R$, uniformly for $X \in \partial \hC_S$, and with a constant $C$ depending only on $d,s$ (because $S > 2R$). Moreover we have $|\nabla g(S)| \leq C S^{-s-1}$, the perimeter of $\partial \hC_S$ is $O(S^{d+k-1})$ and $\yg \leq CS^{\gamma}$ on $\partial \hC_S$. Since $\gamma = s-d+2-k$ we have $S^{\gamma + d+k-1} = S^{s+1}$. 

Combining the estimates above we get for any $S > 2R$
\[ \Big| \int_{\partial \hC_S} \yg \Hloc \Eloc \cdot \vec{n} \Big| \leq C S^{s+1} \left( \D_R^2 \frac{|g(S)|}{S^{s+1}} + (\Nn_R^2 + R^{2d}) \left(\frac{R}{S^{2s+2}} + \frac{R|g(S)|}{S^{s+2}} + \frac{R^2}{S^{2s+3}}\right) \right)
\]
with a constant $C$ depending only on $d,s$.

In the cases \eqref{kernel} the right-hand side  is $O(S^{-s})$ as $S \ti$ and $s > 0$. We thus obtain
\begin{equation} \label{bboundKern}
\Big|\int_{\partial \hC_S}\yg \Hloc \Eloc \cdot \vec{n} \Big| = o_{S \ti}(1).
\end{equation}

In the logarithmic cases we have $s = 0$ and $g(S) = \log S$, hence if we set $S = R^4$ we get
\begin{equation} \label{bboundLog}
\Big| \int_{\partial \hC_{R^4}} \yg \Hloc \Eloc \cdot \vec{n} \Big| \leq \Clog \D_R^2\log R + O\left(\Nn_R^2 R^{-5}\right) + o_R(1),
\end{equation}
for a certain universal constant $\Clog$.

Combining \eqref{energienonneutre} and \eqref{bboundKern} and letting $S \ti$ we obtain, in the cases \eqref{kernel}
\[
\lim_{\eta \t0} \left(\frac{1}{c_{d,s}} \int_{\R^{d+k}} \yg |\Eloc_{\eta}|^2 - \Nn_R g(\eta) \right) = \Hint_R[\C].
\]

Combining \eqref{energienonneutre} and \eqref{bboundLog} we obtain, in the logarithmic cases
\[
\lim_{\eta \t0} \left(\frac{1}{c_{d,s}} \int_{\hC_R} \yg |\Eloc_{\eta}|^2 - \Nn_R g(\eta) \right) \leq \Hint_{R}[\C] + \Clog \D_R^2\log R + O\left(\Nn_R^2 R^{-5}\right) + o_R(1),
\]
where we have used the trival bound $\int_{\hC_R} \yg |\Eloc_{\eta}|^2 \leq \int_{\hC_{R^4}} \yg |\Eloc_{\eta}|^2$.

From \cite[Lemma 2.3]{PetSer} we know that the limit as $\eta \t0$ is almost monotonous, more precisely for any $\eta \in (0,1)$ and $S > 2R$ we have
\begin{equation} \label{almono}
\int_{\hC_S} \yg |\Eloc_{\eta}|^2 - \Nn_R g(\eta) \leq \lim_{\eta \t0} \left(\frac{1}{c_{d,s}} \int_{\hC_S} \yg |\Eloc_{\eta}|^2 - \Nn_R g(\eta) \right)  + \Nn_R o_{\eta}(1).
\end{equation}

We thus obtain \eqref{compkernel} and \eqref{complog}. Inequalities \eqref{relatkern} and \eqref{relatlog} follow easily by taking the expectation under $P$.
\end{proof}

\subsection{Bound below on the interaction}
The self-interaction of an electric system of point charges with a negative background in a given set $K$ is bounded below in terms of the number of points in $K$ (and the volume of $K$ in the logarithmic cases).
\begin{lem} \label{lem:bbint}
Let $K$ be a compact subset such that $K \subset C_R$ and let $\C$ be a point configuration in $K$. We denote by $\Nn(\C,K)$ the number of points in $K$, and we denote by $\Inte[K,K](\C)$ the interaction energy of $K$ with itself as in \eqref{def:inte}.

In the cases \eqref{kernel} we have
\begin{equation} \label{bbkernel}
\Inte[K,K](\C) \geq - C \left(\Nn(\C,K)\right).
\end{equation}

In the logarithmic cases we have
\[
\Inte[K,K](\C) \geq - C \left(\Nn(\C,K) + (\Nn(\C,K) - |K|)^2 \log R + \Nn^2(\C,K) R^{-5}\right).
\]

In both inequalities $C$ is a positive constant depending on $d,s$.
\end{lem}
\begin{proof}
Letting $\Eloc$ be the local electric field generated by $\C$ in $K$ and arguing as in the proof of Lemma \ref{lem:locflocint} we obtain
\[
\Inte[K,K](\C) = \lim_{\eta \t0} \left(\frac{1}{c_{d,s}} \int_{\hC_S} \yg |\Eloc_{\eta}|^2 - \Nn_R g(\eta) \right) \\
+ \frac{1}{\cds} \int_{\partial \hB_S} \yg \Hloc \Eloc \cdot \vec{n}.
\]

As seen in \eqref{almono} the limit as $\eta \t0$ is almost monotonous. In the cases \eqref{kernel}, chosing $\eta = \hal$ and letting $S \ti$ we obtain
\begin{equation} \label{bblog}
\Inte[K,K] \geq \frac{1}{\cds} \int_{\R^{d+k}} \yg |\Eloc_{\hal}|^2 - \Nn(\C, K) (g(1/2) + C), 
\end{equation}
which yields \eqref{bbkernel}. 

In the logarithmic cases, choosing $\eta = \hal$ and taking $S = R^4$ we obtain
\[
\Inte[K,K] \geq \frac{1}{\cds} \int_{\hC_{R^4}} \yg |\Eloc_{\hal}|^2 - \Nn(\C, K) (g(\hal) + C) - \left|\int_{\partial \hC_{R^4}} \yg \Hloc \Eloc_{\hal}\right|.
\]
Controlling the boundary term as in the proof of Lemma \ref{lem:locflocint} yields \eqref{bblog}.
\end{proof}

\subsection{Discrepancy estimates}
The following lemma is an adaptation of the discrepancy estimates of \cite{PetSer} to show how the finiteness of $\Wint(P)$ (instead of $\Welec(P)$) implies that $P$ has a \textit{number variance} of order $o(R^{2d})$.

\begin{lem} \label{lem:discrPkern}
Let $P$ be a stationary random point process of intensity $1$.

In the cases \eqref{kernel} if $\Wint(P)$ is finite then we have
\begin{equation} \label{discrPkern}
\Esp_{P} [\D^2_R] = o(R^{2d}).
\end{equation}

In the logarithmic cases if $\Dlog(P)$ is finite then we also have \eqref{discrPkern}.
\end{lem}
\begin{proof}
In the logarithmic cases, the result is straightforward since $\Dlog(P) < + \infty$ implies that $\Esp_{P} [\D^2_R] = o(R^d)$. 

We now turn to the cases \eqref{kernel}. Applying \cite[Lemma 2.2]{PetSer} with $\eta = \hal$ we get
\[
\frac{1}{R^s} \D^2_R \min\left(1, \frac{\D_R}{R^d}\right) \leq C\int_{\R^{d+k}} \yg |\Eloc_{\hal}|^2 +C \Nn_R.
\]
Taking the expectation under $P$ yields
\[
\Esp_{P} \left[ \frac{1}{R^s} \D^2_R \min\left(1, \frac{\D_R}{R^d}\right) \right] \leq C \Esp_{P} \left[ \int_{\R^{d+k}} \yg |\Eloc_{\hal}|^2\right] + C R^d, 
\]
where $\Eloc$ denotes the local field generated by $\C$ in $C_R$ and $C$ is a constant depending on $d,s$. We know from Lemma \ref{lem:locflocint} that
\[
\int_{\R^{d+k}} \yg |\Eloc_{\hal}|^2 - \Nn_R g(\hal) \leq \Hint_R[\C] + C \Nn_R, 
\]
and combining the previous two estimates we obtain that
\[
\Esp_{P} \left[ \frac{1}{R^s} \D^2_R \min(1, \frac{\D_R}{R^d}) \right] \leq C\Esp_{P}[ \Hint_R] + C R^d = O(R^d).
\]
Using Jensen's inequality and the fact that $s < d$ we get
\[
\Esp_{P} [\D^2_R] = O(R^{\frac{4}{3} d + \frac{2}{3} s}) = o(R^{2d}),
\]
which proves \eqref{discrPkern}.
\end{proof}

\subsection{The screening lemma}
For convenience we recall the “screening lemma” following \cite{SS2d}, \cite{SS1d}, \cite{RougSer}, \cite{PetSer} and \cite{LebSer}. We present it here in a simplified form which will be enough for our purposes, we refer to \cite{LebSer} for the most general statement available and to \cite{PetSer} for a proof.

The result consists in the following: given a point configuration $\C$ in $\carr_R$ and a compatible electric field $E$, we wish to construct another compatible field $\Escr$ such that $\Escr \cdot \vec{n} = 0$ on the boundary of $\carr_R \times \R^k$. Indeed such \textit{screened} fields may be pasted together in adjacent hypercubes because their normal component vanish (and in particular are equal), which will allow us to construct \textit{global} fields defined in the whole space. 

Of course this  construction is not possible in general (e.g. it is easy to see that it imposes a condition on the number of points in $\carr_R$, which must match exactly the volume of $\carr_R$). However, under some conditions on $E$ to be “screenable”, by extending $\carr_R$ a bit and modifiying $\C$ only in a thin layer of width $\epsilon R$ we may find a new point configuration $\Cscr$ and a compatible \textit{screened} field $\Escr$, such that moreover the energy of $\Escr$ is bounded in terms of the energy of $E$.

\begin{lem} \label{lem:screening}
There exists $R_0 > 0$ depending on $d,s$ and $\eta_0 > 0$ depending only on $d$ such that the following holds. 

Let $0 < \epsilon < \frac{1}{2}$ and $0 < \eta < \eta_0$ be fixed.
Let $\carr_R$ be a hypercube of sidelength $R$ for some $R > 0$, and let $K$ be the hypercube of sidelength $\ceil{R}$ (where $\ceil{R}$ denotes the smallest integer larger than $R$).

Assume that $E$ is a vector field defined in $\carr_R \times \R^k$ such that
$$ - \div(\yg E) = c_{d,s} \left(\C- \delta_{\R^d} \right) \text{ in } \carr_R.$$ 

Let $M > 1$ such that $E$ satisfies:
\[
\frac{1}{R^d} \int_{\carr_R \times \R^k} \yg |E_{\eta}|^2 \leq M.
\]

In the case $k=1$ (the non-Coulomb cases) we define $e_{\epsilon,R}$ as
\begin{equation} \label{def:eeps}
e_{\epsilon, R} := \frac{1}{\epsilon^4 R^d} \int_{\carr_R \times (\R \backslash (-\epsilon^2 R, \epsilon^2 R))} \yg |E|^2.
\end{equation}

Under the assumption that the following inequalities are satisfied 
 \begin{multline}\label{condR}
 R> \max \left( \frac{R_0}{\epsilon^{2}},  \frac{R_0 M}{\epsilon^3}\right),   \\ 
R>   \begin{cases}\frac{R_0 M^{1/2}}{\epsilon^{d+3/2}} & \text{if } k =0\\ 
\max(R_0 M^{1/(1-\gamma)} \epsilon^{\frac{-1-2d+\gamma}{1-\gamma}} , R_0 \epsilon^{\frac{2\gamma}{1-\gamma} }  e_{\epsilon,R}^{1/(1-\gamma)}   )  & \text{if} \ k=1,
\end{cases}
\end{multline} 
there exists a point configuration $\Cscr$ in $K$ and a vector field $\Escr \in \Lploc(\R^{d+k},\R^{d+k})$ such that 
\begin{enumerate}
\item The configuration $\Cscr$ has exactly $|K|$ points in $K$.
\item The configurations $\C$ and $\Cscr$ coincide on $\Int_{\epsilon} := \{ x \in \carr_R, \dist(x, \partial \carr_R)\} \geq 2 \epsilon R\}$.
\item We have
 \begin{equation} \label{rendchamp} \left\lbrace\begin{array}{ll}
  -\div \left(\yg \Escr \right) = c_{d,s} \left(\Escr - \delta_{\Rd} \right) & \  \text{in} \ K \times \R^k\\
   \Escr \cdot \vec{n}=0 & \ \text{on} \ \partial K \times \R^k ,\end{array}\right.
 \end{equation} 
\item Letting $\Escr_{\eta}$ be associated to $\Escr$ as in \eqref{def:Eeta} it holds
\begin{multline} \label{erreurEcrantage} \int_{K \times \R^k} \yg |\Escr_{\eta}|^2 \leq \left( \int_{\carr_R \times \R^k} \yg |E_{\eta}|^2 \right) (1+C\ep) \\ + Cg(\eta) M \ep R^d + Ce_{\ep, R} \ep R^d + o(R^{d-1}).
\end{multline}
for some constant $C$ depending only on $s,d$.
\end{enumerate}
\end{lem}

\subsection{Minimality of the local energy}
The following was proven in \cite[Lemma 3.13.]{LebSer}. It expresses the fact that the local electric field $\Eloc$ has a lower energy than any “screened” field compatible with the same point configuration.

\begin{lem} \label{lem:minilocale} 
Let $R \geq 1$ be an integer and $\C$ be a point configuration in $\carr_R$. Let $\Eloc$ be the local electric field generated by $\C$ in $\carr_R$ as in \eqref{def:Eloc}. Let $E \in \Lploc(\R^{d+k}, \R^{d+k})$ be a vector field satisfying 
  \begin{equation}
  \label{checr}
\left\lbrace \begin{array}{ll}
 -\div (\yg E) 
= \cds \left( \C- \drd\right) & \text{in} \ \carr_R \times \R^k \\
 E \cdot \vec{\nu} = 0  & \text{on} \  \partial \carr_R \times \R^k.
 \end{array}
 \right.
\end{equation} 
Then, for any $0<\eta <1$ we have
\begin{equation}\label{comparloc}
\int_{\R^{d+k}} \yg |\Eloc_{\eta}|^2 \leq \int_{\carr_R \times \R^{k}} \yg |E_{\eta}|^2.
\end{equation}
\end{lem}

\section{Connection of the electric and intrinsic approach} \label{sec:connection}
This section is devoted to the proof of Theorem \ref{theo:connec}. It goes in two steps.
\begin{enumerate}
\item First we establish an upper bound $\Welec \leq \Wint$ (or $\Wint + \Dlog$ in the logarithmic cases). The proof of this “electric-intrinsic” inequality is the purpose of Section \ref{sec:UB}.
\item Then, in the non-Coulomb cases, for any $\Pst \in \probas_{s,1}(\config)$ we construct a sequence $\{P_N\}_N$ of stationary random point processes which converges to $\Pst$ and such that $\Welec(\Pst)$ is bounded below by $\lim_{N \ti} \Wint(P_N)$ (we also ensure that $P_N$ is hyperuniform, in particular $\Dlog(P_N)$ is always zero). Moreover we have $\lim_{N \ti} \ERS[P_N|\Poisson] = \ERS[P|\Poisson]$. 

This operation is similar to the construction of a “recovery sequence” in $\Gamma$-convergence, and is proven in Section \ref{sec:recovery}.
\end{enumerate}

These two steps immediatly imply Theorem \ref{theo:connec}.

%

\paragraph{Re-writing of the additional term.}
Using the definition of the two-point correlation function (which exists by assumption as a Radon measure) we get
$$\iint_{\carr_R^2} (\rho_2(x,y) - 1)dx dy = \Esp_{P} \left[ \Nn_R (\Nn_R -1) \right] - R^{2d} = \Esp_{P} \left[ \Nn_R^2 \right] - R^{2d} - \Esp_{P} \left[ \Nn_R \right].
$$
Since $\Pst$ has intensity $1$ we have $\Esp_{P} \left[ \Nn_R \right] = R^d$ hence
\begin{equation} \label{discrasrho2}
\iint_{\carr_R^2} (\rho_2(x,y)-1) dx dy = \Esp_{P} \left[ \D_R^2 \right] - R^d.
\end{equation}
Equation \eqref{discrasrho2} allows us to write the term $\Dlog$ (defined in \eqref{def:Dlog}) in terms of $\rho_2$ equivalently as
\begin{multline} \label{rewriteDlog}
\Dlog(P) = \Clog \limsup_{R \ti} \left( \frac{1}{R^d}  \iint_{\carr_R^2} (\rho_2(x,y) - 1)dx dy + 1 \right) \log R \\
= \Clog \limsup_{R \ti}  \frac{1}{R^d} \Esp_{P} \left[ \D_R^2 \log R \right].
\end{multline}

\subsection{The electric-intrinsic inequality} \label{sec:UB}
\emph{Until the end of Section \ref{sec:UB}, $P$ denotes a stationary random point process of intensity $1$ on $\Rd$ such that $\Wint$ is finite. In the logarithmic cases we assume that $\Dlog(P)$ is finite.}

The following proposition is the first part of the proof of Theorem \ref{theo:connec}.
\begin{prop} \label{prop:elecint}
Under the above assumptions, we have in the cases \eqref{kernel}
\[
\Welec(\Pst) \leq  \Wint(\Pst) ,
\]
and in the logarithmic cases $\Welec(\Pst) \leq \Wint(\Pst) + \Dlog(\Pst)$.
\end{prop}

\paragraph{Screenability of the local electric fields.}
Let $\eta, \epsilon > 0$ be fixed. Let $\{R_n\}_n$ be an increasing sequence of real numbers such that $\lim_{n \ti} R_n = + \infty$ and
\begin{equation} \label{def:Rn}
\lim_{n \ti} \frac{1}{R_n^d} \Esp_{P} [ \Hint_{R_n} ] = \Wint(P).
\end{equation}
We start by an auxiliary lemma. 
\begin{lem} \label{lem:screenability}
Let $\Elocn$ denotes the local electric field generated by a point configuration in $\carr_{R_n}$. The following inequality holds
\begin{equation}
\label{Mpetit} P \left( \frac{1}{R_n^d} \int_{\carr_{R_n} \times \R^k} \yg |\Elocn_{\eta}|^2 \leq M \right) = 1 - O(M^{-1}),
\end{equation}
moreover, in the case $k=1$,
\begin{equation}
\label{epetit} P\left(\frac{1}{\epsilon^4 R_n^d} \int_{\carr_{R_n} \times (\R \backslash(-\epsilon^2 R_n, \epsilon^2 R_n))} \yg |\Elocn|^2 \leq \left(10 R_0\epsilon^{\frac{2\gamma}{1-\gamma} }\right)^{\gamma-1} R_n^{d-s} \right) = 1 - o_n(1).
\end{equation}
\end{lem}
\begin{proof}
The first point \eqref{Mpetit} follows directly from Markov's inequality and the assumptions on $P$. Indeed we have, in the cases \eqref{kernel} (using \eqref{relatkern})
\[
\Esp_{P} \left[ \frac{1}{\cds} \int_{\carr_{R_n} \times \R^k} \yg |\Elocn_{\eta}|^2 \right] \leq \Esp_{P} [ \Hint_{R_n} ] + O(R_n^d),
\]
and in the logarithmic cases (using \eqref{relatlog})
\[
\Esp_{P} \left[ \frac{1}{\cds} \int_{\carr_{R_n} \times \R^k} \yg |\Elocn_{\eta}|^2 \right] \leq \Esp_{P} [ \Hint_{R_n}] + \Clog\Esp_{P} [ \D^2_{R_n}] \log R_n + O(R_n^d),
\]
where the terms $O(R_n^d)$ depend on $\eta$. 

To prove the second point \eqref{epetit} let us fix $X = (x,y)$ with $x \in \carr_{R_n}$ and $y \geq \epsilon^2 R_n$. We may estimate
$\Elocn(X)$ as follows: let $R_0 > 0$ and let us divide $R_n$ into $O(R_n^d/R_0^d)$ hypercubes $\{\bC_i\}_{i \in I}$ of sidelength $\in (\hal R_0, \frac{3}{2} R_0)$. For any $i \in I$ we have by a mean value argument
\[
\left|\int_{\bC_{i}} \nabla g(X-t) (d\C(t) - dt)\right| \leq C|\D_i| |y|^{-s-1} + CR_0(\Nn_i + R_0^d)|y|^{-s-2},
\]
where $\D_i$ (resp. $\Nn_i$) denotes the discrepancy (resp. the number of points) in $\bC_i$. Summing over $i \in I$ we get
\[
|\Elocn(X)| \leq C \sum_{i \in I} |\D_i| |y|^{-s-1} + C R_0 (\Nn_{R_n} + R_n^{d}) |y|^{-s-2}.
\]
We may thus write, using the stationarity of $P$
\[
\Esp_{P} [|\Elocn(X)|^2] \leq C \frac{R_n^{2d}}{R_0^{2d}} \Esp_{P} [\D^2_{R_0}] |y|^{-2s-2} + CR_0 (\Esp[\Nn_{R_n}^2] + R_n^{2d}) |y|^{-2s-4}.
\]
Using Lemma \ref{discrPkern} or the fact that $\Dlog(P)$ is finite we have $\Esp[\Nn_{R_n}^2] = R_n^{2d} + O(R_n^{2d})$. Finally, integrating over $\carr_R \times \R \backslash(-\epsilon^2 R_n, \epsilon^2 R_n))$ we obtain
\begin{multline*}
\Esp_{P} \left[ \frac{1}{\epsilon^4 R_n^d} \int_{\carr_{R_n} \times (\R \backslash(-\epsilon^2 R_n, \epsilon^2 R_n))} \yg |\Elocn|^2  \right] \\ \leq \frac{1}{\epsilon^4} \left( C \Esp_P[\D_{R_0}^2] \frac{R_n^{2d}}{R_0^{2d}} \frac{1}{(\epsilon^2 R_n)^{d+s}} + C R_0 R_n^{2d} \frac{1}{(\epsilon^2 R_n)^{d+s+2}}\right).
\end{multline*}
Using the fact that $\frac{1}{R_0^{2d}} \Esp_{P}[\D_{R_0}^2] = o(1)$ as $R_0 \ti$ (see Lemma \ref{lem:discrPkern}), we obtain
\[
\Esp_{P} \left[ \frac{1}{\epsilon^4 R_n^d} \int_{\carr_{R_n} \times (\R \backslash(-\epsilon^2 R_n, \epsilon^2 R_n))}  \yg |\Elocn|^2 \right]= o(R_n^{d-s}).
\]
The bound \eqref{epetit} follows by Markov's inequality.
\end{proof}

We now turn to the proof of Proposition \ref{prop:elecint}.
\begin{proof}

\textbf{Screening the local electric fields.}
When $\eta, \epsilon > 0$ and $M > 0$ are fixed we denote by 
$\Spn$ the set of point configurations in $\carr_{R_n}$ such that $\Elocn$ satisfies \eqref{Mpetit} and, in the case $k=1$, \eqref{epetit}. For any $\C$ in $\Spn$ the conclusions of Lemma \ref{lem:screening} apply to $\Elocn$. We may thus find a point configuration $\Cscr$ in $K_n := \carr_{\ceil{R_n}}$ and a compatible field $\Escr$ such that
\begin{enumerate}
\item The point configurations $\C$ and $\Cscr$ coincide on a large subset of $C_{R_n}$, namely $\{x \in C_{R_n}, \dist(x, \partial C_{R_n} \geq 2\epsilon R_n\}$.
\item The vector field $\Escr$ is screened i.e. $\Escr \cdot \vec{n} = 0$ on $\partial K_n \times \R^k$.
\item The energy of $\Escr$ is bounded in terms of that of $\Elocn$ as in \eqref{erreurEcrantage}
\begin{multline} \label{erreurscr}
\int_{K_n \times \R^k} \yg |\Escr_{\eta}|^2 \leq \left( \int_{\carr_{R_n} \times \R^k} \yg |\Elocn_{\eta}|^2 \right) (1+C\ep) \\ + Cg(\eta) M \ep R_n^d + Ce_{\ep, R} \ep R_n^d + o(R_n^{d-1}),
\end{multline}
where $e_{\ep,R}$ is defined in \eqref{def:eeps}.
\end{enumerate}

\textbf{Constructing a global electric field.}
Any such point configuration (resp. vector field) may be extended periodically in the whole space $\R^d$. The main point is that since $\Escr$ is screened we can paste together several copies of $\Escr$ periodically without creating divergence at the boundary of two tiles. Let $\Cper$ (resp. $\Eper$) be the resulting periodic point configuration (resp. vector field)

We let $P_n$ be the conditional expectation of $P$ knowing $\Spn$ and we let $\Pper_n$ be the push-forward of $P_n$ by the map $\C \mapsto \Cper$ defined above. Finally we let $\Pav_n$ be the average of $\Pper_n$ over translations in $K_n$. Taking the expectation of \eqref{erreurscr} under $P$ we see that (with $\Welec_{\eta}$ as defined in \eqref{def:Welec})
\begin{multline*}
\Welec_{\eta}(\Pav_n) \leq \frac{1}{R_n^d} \Esp_{P_n} \left[\frac{1}{\cds} \int_{K_n \times \R^k} \yg |\Escr_{\eta}|^2 \right] - g(\eta) \\
\leq \frac{1}{R_n^d} \Esp_{P_n} \left[\frac{1}{\cds} \int_{\carr_{R_n} \times \R^k} \yg |\Elocn_{\eta}|^2 \right] (1+C\ep)  - g(\eta) + Cg(\eta) M \ep \\
+ C \ep \Esp_{P_n} [e_{\ep, R_n}] + o(R_n^{-1}).
\end{multline*}

From Lemma \ref{lem:screenability} we see that as $M \ti$, $n \ti$ the random point process $P_n$ converges to $P$. In particular we may bound the expectations under $P_n$ in the right-hand side by the expectation under $P$ at a small cost
\[
\Esp_{P_n} \left[\int_{\carr_{R_n} \times \R^k} \yg |\Elocn_{\eta}|^2 \right] \leq \Esp_{P} \left[\int_{\carr_{R_n} \times \R^k} \yg |\Elocn_{\eta}|^2 \right] (1 + o(1)),
\]
\[
\Esp_{P_n} [e_{\ep, R_n}] \leq \Esp_{P} [e_{\ep, R}](1+o(1))
\]
where both terms are $o(1)$ as $M \ti, n \ti$ (keeping $\epsilon, \eta$ fixed). By definition of $e_{\ep,R}$ and a mean value argument we see that up to changing $\epsilon$ into $2\epsilon$ we may assume that
\[
\Esp_{P} [e_{\ep, R_n}] = \Esp_{P}\left[ \frac{1}{\epsilon^4 R_n^d} \int_{\carr_{R_n} \times (\R \backslash(-\epsilon^2 R_n, \epsilon^2 R_n))} \yg |\Elocn|^2\right] \leq \frac{1}{\epsilon^6 R_n} M.
\]
We thus get
\begin{multline*}
\Welec_{\eta}(\Pav_n) \leq \frac{1}{R_n} \Esp_{P} \left[\frac{1}{\cds} \int_{\carr_{R_n} \times \R^k} \yg |\Elocn_{\eta}|^2 \right] (1 + o(1)) (1+C\ep) -g(\eta) \\ + Cg(\eta) M \ep
+ C \ep^{-5} R_n^{-1}.
\end{multline*}
The convergence of $P_n$ to $P$ implies the convergence of $\Pav_n$ to $P$ as $M, n \ti$ and $\epsilon \t0$. This might be seen as follows: the topology on $\probas(\config)$ is such that for any fixed $\delta > 0$, if two random point processes coincide (or are close to each other) in $C_{R_0}$ for $R_0$ large enough, then they are $\delta$-close. In particular $P_n$ and $P$ are very close to each other in $\carr_{R_n}$ because on the one hand the vast majority of point configurations under $P$ are screenable and on the other hand the screening procedure does not modify the points in a large interior part of $\carr_{R_n}$. Heuristically speaking, if $R_n$ is larger than some $R_0$ then $P_n$ and $P$ should be $\delta$-close. Now when averaging the random point process $P_n$ over translations in $\carr_{R_n}$ there is a nonzero proportion of $z \in \carr_n$ which are such that the translation by $z$ of the thin layer of $\carr_{R_n}$ in which the points have been modified ends up intersecting $\carr_{R_0}$, thus $\Pav_n$ look \textit{less} like $P$ in $\carr_0$. However when $R_n \gg R_0$ this proportion is of order $\epsilon$, hence $\Pav_n$ still converges to $P$ when taking $M,n$ large and $\epsilon$ small.

\textbf{Conclusion.}
Taking $M$ large, letting $n \ti$ and using \eqref{relatkern} or \eqref{relatlog} we obtain, by lower semi-continuity of $\Welec_{\eta}$ over $\probas_{s,1}(\config)$, (see Lemma \ref{lem:lsciW}) that
\[
\Welec_{\eta}(P) \leq \Wint(P)(1+C\ep) + O(\ep) + o_{\eta}(1),
\]
plus an additional $\Dlog(P)$ term in the logarithmic cases. Sending $\epsilon \t0$ and $\eta \t0$ we finally obtain that
$\Welec(P) \leq \Wint(P)$ or, in the logarithmic cases $\Welec(P) \leq \Wint(P) + \Dlog(P)$, which concludes the proof of Proposition \ref{prop:elecint}.
\end{proof}

\subsection{Construction of a recovery sequence} \label{sec:recovery}
\textit{In this section $P$ denotes a stationary random point process of intensity $1$ such that $\Welec(P)$ is finite. Moreover we assume that we are in one of the non-Coulomb cases, i.e. $d=1$ or $d \geq 2$ and $s > d-2$.}

The following result forms the second step in the proof of Theorem \ref{theo:connec}.
\begin{prop} \label{prop:partialconv} There exists a sequence $\{P_{N}\}_{N}$ of \textit{hyperuniform} random point processes in $\probas_{s,1}(\config)$ such that 
\[
\lim_{N \ti} P_N = P, \quad \lim_{N \ti} \ERS[P_N|\Poisson] = \ERS[P|\Poisson],
\]
and satisfying
\begin{equation} 
\label{convWint} \lim_{N \ti} \Wint(P_N) = \Welec(P).
\end{equation}
Let us observe that since the random point processes are hyperuniform, in the logarithmic cases they satisfy $\Dlog(P_N) = 0$ for any $N$.
\end{prop}

The proof of Proposition \ref{prop:partialconv} goes in two steps. 
\begin{enumerate}
\item First we construct an auxiliary sequence of random point processes which converges to $P$ and such that almost every point configuration is finite and “screened” i.e. there exists an associated \textit{screened} electric field. This is done in Lemma \ref{lem:auxseq}.
\item Next, we extend this random point process in the whole space and make it stationary, before bounding its interaction energy in terms of $\Welec(P)$.
\end{enumerate}
 
\subsubsection{An auxiliary sequence}
\begin{lem} \label{lem:auxseq}
There exists a sequence $\{P^{(1)}_N\}_{N}$ of random point processes in $\probas(\config(C_N))$ such that

\textbf{0.} The sequence $\{P^{(1)}_N\}_{N}$ converges to $P$ as $N \ti$. More precisely, there exists a sequence $\{L_N\}_N$ such that $L_N = N(1-o(1))$ and  the respective restrictions of $\Pmod_N$ and $P$ to $C_{L_N}$ are arbitrarily close as $N \ti$. 

\textbf{1.} For $P^{(1)}_N$-a.e. point configuration $\C^{(1)}$ there exists a \textit{screened} electric field $E^{(1)}$ satisfying
\begin{equation} \label{Ecompat}
- \div(\yg E^{(1)}) = \cds(\C^{(1)} - \delta_{\R^d}) \text{ in } \carr_N \times \R^k,
\end{equation}
\begin{equation} \label{Escrscr}
E^{(1)} \cdot \vec{n} = 0 \text{ on }\partial \carr_N \times \R^k.
\end{equation}
In particular the point configurations have $P^{(1)}_N$-a.s. $N^d$ points in $\carr_N$. We also have
\begin{equation} \label{distanceaubord}
\min_{p \in \C^{(1)}} \dist(p, \partial C_N) \geq \eta_0,
\end{equation}
for some $\eta_0 >0$ depending only on $d,s$.

\textbf{2.} The following estimate holds
\begin{equation} \label{Escrener}
\limsup_{N \ti} \lim_{\eta \t0}  \Esp_{P^{(1)}_N} \left[ \frac{1}{\cds} \frac{1}{N^d} \int_{\carr_N \times \R^k} \yg  |E^{(1)}_{\eta}|^2 - g(\eta) \right] \leq \Welec(P).
\end{equation}

\textbf{3.} The relative entropies of $P^{(1)}_N$ and $P_{|\carr_N}$ with respect to $\Poisson_{|\carr_N}$ are close
\begin{equation} \label{convEnt}
\Ent[P^{(1)}_N |\Poisson_{|\carr_N}] = \Ent[P_{|\carr_N} |\Poisson_{|C_N}] + o(N^d).
\end{equation}
\end{lem}
\begin{proof}
This follows from the analysis of \cite{LebSer}, and we sketch here the main steps.

Let $\Pelec$ be a stationary electric process associated to $P$ as in Lemma \ref{lem:concordWelec}. For fixed $R, M, \epsilon, \eta > 0$ we say that an electric field $E$ is in $\Sp$ (or is \textit{screenable}) if its energy is controlled as follows
\[
\frac{1}{R^d} \int_{\carr_R \times \Rd} \yg |E_{\eta}|^2 \leq M \text{ and, if $k =1$, } \frac{1}{\epsilon^4 R^d} \int_{\carr_R \times \Rd \backslash (-\epsilon^2 R, \epsilon^2 R)} \yg |E_{\eta}|^2 \leq 1.
\]
Under the assumption that $\Welec(P)$ is finite, then the probability $\Pelec(\Sp)$ tends to $1$ as $M, R \ti$ for any $\epsilon, \eta >0$ fixed. This is proven in \cite[Lemma 5.10]{LebSer} and is similar in spirit to Lemma \ref{lem:screenability}.

If $R$ is an integer much larger than $M$ and $E$ is in $\Sp$, the screening procedure as in Lemma~\ref{lem:screening} (see also \cite[Proposition 6.1]{PetSer} and \cite[Proposition 5.2]{LebSer}) applies. In particular we may change the underlying point configuration in a thin layer of size $\leq \epsilon R$ close to the boundary of $\carr_R$ and obtain a new \textit{screened} point configuration $\Cscr$ in $\config(\carr_R)$ as well as a compatible \textit{screened} electric field $\Escr$ which satisfies \eqref{Escrscr}. It also ensures that \eqref{distanceaubord} holds. The screening procedure is described in \cite[Section 5.1]{LebSer}. The energy of $\Escr$ is bounded in terms of the energy of $E$ as in \eqref{erreurEcrantage}.

The next step is to regularize the point configurations $\Cscr$ by separating the pair of points which are close from each other.  This regularization procedure is described in \cite[Section 5.2]{LebSer}, and another electric field $\Emod$ can be associated to the regularized point configurations, with a good energy bound. The main benefit of this procedure is to control the difference between $\frac{1}{R^d} \int_{\carr_R \times \R^k} \yg |\Emod_{\eta}|^2 - g(\eta)$ and its limit as $\eta \t0$. In general the limit may be much larger because of the contribution of pairs of points which are very close, at distance $\ll \eta$, and which are not “seen” when truncating at scale $\eta > 0$. 

We let $\epsilon, \eta$ tend to $0$ and $M$ tend to infinity (depending on $N$) and we pick $R = N$ large enough. We let $P^{(1)}_N$ be the associated random point process in $\carr_N$. Most of the point configurations (or electric fields) are “screenable”, the screening procedure only modifies the configuration in a thin boundary layer of $\carr_N$, and the regularization moves only a \textit{small} fraction of the points by a \textit{small} distance. This ensures that $\P^{(1)}_N$ converges to $P$ as $N \ti$. 

The estimates on the energy of the screened-then-regularized electric fields are such that \eqref{Escrener} holds (see \cite[Section 6.3.4.]{LebSer}) with $E^{(1)} := \Emod$.

Concerning the entropy, letting the new/modified points of the configurations move in small balls allow us to recover a small, nonzero volume in phase space without affecting the energy, since only a small fraction of the points have been deleted/created/modified, it gives \eqref{convEnt} (see \cite[Section 6.3.5]{LebSer} for a precise analysis of the volume loss).
\end{proof}

\subsubsection{Proof of Proposition \ref{prop:partialconv}}
\begin{proof}
\textbf{Step 1.} \textit{Construction of the random point process.}
Let $\{P^{(1)}_N\}_N$ be as in Lemma \ref{lem:auxseq} and let us extend $P^{(1)}_N$ in the whole space $\R^d$ as follows. Let $\{\bC_i\}_{i \in I}$ be a tiling of $\R^d$ by a family of hypercubes of sidelength $N$ and let $x_i$ be the center of $x_i$ (we may impose that one of the $x_i$'s is $0$). Let $\{P^{(1)}_{N,i}\}_{i \in I}$ be the laws of independent point processes distributed as $P^{(1)}_{N}$. 

To any family $\{\C^{(i)}\}_{i \in I}$ of point configurations in $C_N$ we may associate the point configuration 
\[
\C := \sum_{i \in I} \theta_{x_i} \cdot \C^{(i)}
\]
which amounts to “paste” the point configuration $\C^{(i)}$ in the hypercube $\bC_i$. 

For any $i \in I$, let $E^{(i)}$ be an electric field which is compatible with $\C^{(i)}$ as in \eqref{Ecompat} and screened as in \eqref{Escrscr}. By the latter condition, the normal component of each $E^{(i)}$ vanishes on the boundary of $C_N \times \R^k$, thus we may paste together such fields along the boundaries and their energy is additive. In particular the electric field $E$ defined by
\begin{equation} \label{def:Econst}
E(x) := \sum_{i \in I}  E^{(i)}(x-x_i)
\end{equation}
is compatible with $\C$ and moreover we have
\[
\int_{A} |E_{\eta}|^2 = \sum_{i \in I} \int_{A} |E^{(i)}_{\eta}|^2
\]
for any measurable subset $A \subset \R^{d+k}$ and any $\eta > 0$. Let us also observe that, by construction, the normal component of $E$ vanishes on the boundary of $C_{mN} \times \R^k$ for any $m \geq 1$ (in fact it is easy to see that it vanishes on any path included in $N\Z^d \times \R^k$).

We let $P^{(2)}_{N}$ be the random point process obtained by pasting $P^{(1)}_{N,i}$ in $\bC_i$ for $i \in I$, or in other terms the push-forward of the product measure of the $\{P^{(1)}_{N,i}\}_{i \in I}$ by the map 
\[
\{\C^{(i)}\}_{i \in I} \mapsto \C := \sum_{i \in I} \theta_{x_i} \cdot \C^{(i)}.
\]
For any $z \in C_{N}$, we let $P^{(2)}_{N,z} = \theta_z \cdot P^{(2)}_N$ be the push-forward of $P^{(2)}_N$ by the translation by a vector $z$. We also define 
$P^{(3)}_{N}$ as the uniform average of $P^{(2)}_{N,z}$ for $z \in C_N$. It is not hard to check that $P^{(3)}_N$ is a hyperuniform stationary random point process which converges to $P$ as $N \ti$.

\textbf{Step 2.} \textit{Estimates on the energy.}
We fix $z \in C_{N}$ and $m \geq 1$, and we use the subscript $z$ to denote a translation by $z$, e.g. $C_{mN, z}$ denotes the hypercube $C_{mN}$ translated by $z$.

By construction, to $P^{(2)}_{N,z}$-almost every point configuration in $C_{(m+1)N,z}$ we may associate an electric field $E$ whose normal component vanishes on the boundary of $C_{(m+1)N,z} \times \R^k$.

Let us recall that $\Inte[A, B](\C)$ denotes the interaction energy between the sets $A$ and $B$ (cf. \eqref{def:inte}). By minimality of the local energy (cf. Lemma \ref{lem:minilocale}) we have for $P^{(2)}_{N,z}$-almost every $\C$
\[
\Inte[C_{(m+1)N,z}, C_{(m+1)N,z}] (\C) \leq \lim_{\eta \t0} \left( \frac{1}{\cds} \int_{C_{(m+1)N,z} \times \R^k} |E_{\eta}|^2 - N^d g(\eta)\right).
\]
Using \eqref{Escrener} we see that the right-hand side is bounded in terms of $\Welec(P)$ as $N \ti$, more precisely we obtain
\begin{equation} \label{energCmpN}
\limsup_{N \ti} \limsup_{m \ti} \frac{1}{(mN)^d} \Esp_{P^{(2)}_{N,z}} \left[ \Inte[C_{(m+1)N,z}, C_{(m+1)N,z}] \right] \leq \Welec(P),
\end{equation}
and both limits (as $m \ti$ and $N \ti$) are uniform for $z \in \carr_{N}$.  

This gives an asymptotic upper bound on the expectation of $\Inte[ C_{(m+1)N,z}, C_{(m+1)N,z}]$ under $P^{(2)}_{N,z}$, however the relevant quantity to control in order to get \eqref{convWint} is rather the expectation of
$\Inte[C_{mN}, C_{mN}]$.
We thus need to bound the difference  $\Inte[ C_{(m+1)N,z}, C_{(m+1)N,z}] -  \Inte[C_{mN}, C_{mN}]$.

Let us write $\Inte[C_{(m+1)N,z}, C_{(m+1)N,z}]$ as
\begin{multline} \label{decompoCMp1}
\Inte[ C_{(m+1)N,z}, C_{(m+1)N,z}] = \Inte[C_{mN}, C_{mN}] + 2 \Inte[C_{mN}, C_{(m+1)N,z} \backslash C_{mN}] \\ + \Inte[C_{(m+1)N,z} \backslash C_{mN}, C_{(m+1)N,z} \backslash C_{mN}]. 
\end{multline}

We may bound the last term in the right-hand side of \eqref{decompoCMp1} using Lemma \ref{lem:bbint} with (for the notations of the lemma) $K = C_{(m+1)N,z} \backslash C_{mN}$ and $R = (m+1)N$. In $C_{(m+1)N,z} \backslash C_{mN}$ there are $O(m^{d-1})$ points, and the discrepancy between the number of points and the volume is also such that $(\Nn(\C, K) - |K|)^2 = O(m^{d-1})$. Applying \eqref{bbkernel} in the cases \eqref{kernel} and \eqref{bblog} in the logarithmic case we obtain
\begin{equation} \label{intelastterm}
\Inte[C_{(m+1)N,z}\backslash C_{mN}, C_{(m+1)N,z} \backslash C_{mN}] \geq o(m^d) \text{ as $m \ti$. }
\end{equation}
We are left to estimate the interaction term $\Inte[C_{mN}, C_{(m+1)N,z} \backslash C_{mN}]$. We may split it as 
\begin{multline} \label{decompo2}
\Inte[C_{mN}, C_{(m+1)N,z} \backslash C_{mN}] = \Inte[C_{(m-1)N,z}, C_{(m+1)N,z} \backslash C_{mN}] \\ + \Inte[C_{mN} \backslash C_{(m-1)N,z}, C_{(m+1)N,z} \backslash C_{mN}],
\end{multline}
and we will prove 
\begin{equation} \label{controledecompo2}
\Esp_{P^{(2)}_{N,z}} \left[ \Inte[C_{mN}, C_{(m+1)N,z} \backslash C_{mN}] \right] = o(m^d).
\end{equation}

\subparagraph{Studying $\Inte[C_{(m-1)N,z}, C_{(m+1)N,z} \backslash C_{mN}]$.}
First we claim that 
\begin{equation} \label{CS}
\Esp_{P^{(2)}_{N,z}} \left[ \Inte[C_{(m-1)N,z}, C_{(m+1)N,z} \backslash C_{mN}] \right]  = o(m^d).
\end{equation}
To prove \eqref{CS} let us write
\[
\Inte\left[C_{(m-1)N,z}, C_{(m+1)N,z} \backslash C_{mN} \right] = - \frac{1}{\cds} \int_{C_{(m-1)N,z}} \Phi^a \div(\yg E^b),
\]
where $\Phi^a$ is the local electric potential generated by the system of charges in $C_{(m+1)N,z} \backslash C_{mN}$ i.e.
\[
\Phi^a(x) := \int_{C_{(m+1)N,z} \backslash C_{mN}} g(x-t) (d\C(t) - dt)
\]
and $E^b$ is the screened electric field associated to the system of charges in $C_{(m-1)N,z}$. Using \eqref{distanceaubord} we may also write
\[
\Inte\left[C_{(m-1)N,z}, C_{(m+1)N,z} \backslash C_{mN} \right] = - \frac{1}{\cds} \int_{C_{(m-1)N,z}} \Phi^a_{\eta_0} \div(\yg E^b_{\eta_0}),
\]
because the minimal distance between a point charge in $C_{(m-1)N,z}$ and one in $C_{(m+1)N,z} \backslash C_{mN}$ is $\geq \eta_0$. An integration by parts and Cauchy-Schwarz's inequality yield (observing also that $C_{(m-1)N,z} \subset C_{(m+2)N}$)
\begin{multline} \label{utilisCS}
\Inte\left[C_{(m-1)N,z}, C_{(m+1)N,z} \backslash C_{mN} \right] \\ \geq - C \left(\int_{C_{(m+2)N} \times \R^k}  \yg |E^a_{\eta_0}|^2\right)^{1/2} \left(\int_{\R^{d+k}} \yg |E^b_{\eta_0}|^2 \right)^{1/2}, 
\end{multline}
for some constant $C$ depending on $d,s$. Using the definition of $P^{(2)}_{N,z}$ we have
\[
\Esp_{P^{(2)}_{N,z}} \left[ \int_{\R^{d+k}} |E^b_{\eta_0}|^2 \right] \leq m^d \Esp_{P^{(2)}_{N}} \left[ \int_{C_N \times \R^k} |E^{(1)}_{\eta_0}|^2 \right],
\]
where $E^{(1)}$ is as in \eqref{Ecompat}. Using \eqref{Escrener} we obtain that
\[
\limsup_{N \ti, m \ti}  \Esp_{P^{(2)}_{N,z}} \left[ \frac{1}{(mN)^d} \int_{\R^{d+k}} |E^b_{\eta_0}|^2  \right] \leq ( \Welec(P) + C)
\]
with $C$ depending only on $d,s$. In particular this yields
\[
\Esp_{P^{(2)}_{N,z}} \left[ \int_{\R^{d+k}} |E^b_{\eta_0}|^2 \right] = O(m^d).
\]
Using Jensen's inequality we may thus bound the second term in the right-hand side of \eqref{utilisCS} as
\begin{equation}\label{continter}
\Esp\left[ \left(\int_{\R^{d+k}} \yg |E^b_{\eta_0}|^2 \right)^{1/2} \right] = O(m^{d/2}).
\end{equation}

It remains to control the the first term in the right-hand side of \eqref{utilisCS}. We claim that
\begin{equation} \label{contrbordbord}
\Esp_{P^{(2)}_{N,z}} \left[ \int_{C_{(m+2)N} \times \R^k}  \yg |E^a_{\eta_0}|^2 \right] = o(m^{d}).
\end{equation}
Since $E^a$ is the local field generated by the configuration in $C_{(m+1)N,z} \backslash C_{mN}$ we may, by almost monotonicty as in Lemma \ref{lem:locflocint} compare $\int_{C_{(m+2)N} \times \R^k}  \yg |E^a_{\eta_0}|^2$ with the interaction energy 
\[
\Int[C_{(m+1)N,z} \backslash C_{mN}, C_{(m+1)N,z} \backslash C_{mN}],
\]
up to a boundary term in the logarithmic case (which is bounded as usual) and a term of order $O(m^{d-1}) g(\eta_0) = o(m^d)$. We now prove that 
\begin{equation} \label{espbordbord}
\Esp_{P^{(2)}_{N,z}} \left[\Int[C_{(m+1)N,z} \backslash C_{mN}, C_{(m+1)N,z} \backslash C_{mN}] \right] = o(m^d).
\end{equation}

Let us recall that $\{\bC_i\}_{i \in I}$ denotes a tiling of $\R^d$ by a family of hypercubes of sidelength $N$. 
We let $J_{N,z}$ be the set of indices
\[
J_{N,z} := \{j \in I, (\bC_{j,z} \cap C_{(m+1)N,z} \backslash C_{mN}) \neq \emptyset \}.
\]
and it is clear that the cardinal of $J_{N,z}$ is $O(m^{d-1})$. We may then write the interaction of $C_{(m+1)N,z} \backslash C_{mN}$ with itself as
\begin{multline*}
\Int\left[C_{(m+1)N,z} \backslash C_{mN}, C_{(m+1)N,z} \backslash C_{mN}] \right] \\ = - \sum_{j_1 \neq j_2 \in J_{N,z}} \int_{(\bC_{j_1, z} \cap C_{(m+1)N,z} \backslash C_{mN}) \times (\bC_{j_2,z} \cap C_{(m+1)N,z} \backslash C_{mN})} g(x-y) (d\C(x) - dx) (d\C(y) - dy), \\
+ \sum_{j \in J_{N,z}} \Inte[\bC_{j,z} \cap C_{(m+1)N,z} \backslash C_{mN}), \bC_{j,z} \cap C_{(m+1)N,z} \backslash C_{mN}].
\end{multline*}
The previous identity is nothing but writing the interaction of a collection of (possibly truncated) hypercubes with itself as the sum of hypercubes-hypercubes interactions plus the sum of self-interactions.

Since there are $O(m^{d-1})$ elements in $J_{N,z}$ we have
\[
\Esp_{P^{(2)}_{N,z}} \left[\sum_{j \in J_{N,z}} \Inte[\bC_{j,z} \cap C_{(m+1)N,z} \backslash C_{mN}), \bC_{j,z} \cap C_{(m+1)N,z} \backslash C_{mN}]\right] = O(m^{d-1}).
\]
It remains to bound the sum of interactions between two disjoint hypercubes. We have, if the two hypercubes are not adjacent (since there are only  $O(m^{d-1})$ pairs of adjacent hypercubes in the sum, all giving a contribution of order $O(1)$, we may neglect these terms),
\begin{multline*}
\int_{(\bC_{j_1, z} \cap C_{(m+1)N,z} \backslash C_{mN}) \times (\bC_{j_2,z} \cap C_{(m+1)N,z} \backslash C_{mN})} g(x-y) (d\C(x) - dx) (d\C(y) - dy) \\
 \geq -C  g(\dist(\bC_{j_1}, \bC_{j_2})),
\end{multline*}
with a constant $C$ depending on $N$ and $d,s$. For any fixed $j \in J_{N,z}$ we have (the sum is implicitely restricted to non-adjacent hypercubes)
\[
\sum_{j' \neq j \in J_{N,z}}  g(\dist(\bC_{j}, \bC_{j'})) \leq C \int_1^{mN} r^{-s} r^{d-2} dr \leq C(m^{d-1-s} +1),
\]
with a constant $C$ depending on $N$ and on $d,s$. The element of volume is only $r^{d-2}$ because we are summing terms on the \textit{boundary} of $C_{mN}$. We may then estimate the sum of pairwise hypercube interactions as
\[
\sum_{j_1 \neq j_2 \in J_{N,z}}  g(\dist(\bC_{j_1}, \bC_{j_2})) \leq O(m^{d-1})  (O(m^{d-1-s}) + O(1)) = o(m^d),
\]
because we are restricted to the non-Coulomb cases for which $s > d-2$. It proves \eqref{espbordbord} hence also \eqref{contrbordbord}. Combining \eqref{contrbordbord} and \eqref{continter} we obtain \eqref{CS}. 

\subparagraph{Study of $Inte[C_{mN} \backslash C_{(m-1)N,z}, C_{(m+1)N,z} \backslash C_{mN}]$.}
It remains to control the last term in the right-hand side of \eqref{decompo2}, which is another boundary-boundary interaction and is bounded as in \eqref{espbordbord}.

\subparagraph{Conclusion for the energy.}
Finally, combining \eqref{intelastterm} and \eqref{controledecompo2} we obtain that 
\[
\Esp_{P^{(2)}_{N,z}} \left[ \Inte[ C_{(m+1)N,z}, C_{(m+1)N,z}] -  \Inte[C_{mN}, C_{mN}] \right] \geq o(m^d), 
\]
with a $o(m^d)$ depending on $N$ and $d,s$, but uniform in $z \in C_N$. Combining this estimate with \eqref{energCmpN} we thus conclude that
\[
\limsup_{N \ti} \limsup_{m \ti} \Esp_{P^{(3)}_{N}} \left[\Inte[C_{mN}, C_{mN}] \right] \leq \Welec(P).
\]

\textbf{Step 3.} \textit{Entropy and conclusion.}
The convergence of the entropy follows easily from \eqref{convEnt} and the definition of the relative specific entropy, indeed we have
\[
\ERS[P^{(3)}_N|\Poisson] = \frac{1}{N^d} \Ent[P_N |\Poisson_{\carr_N}] + o(1).
\]

To summarize, we have shown that the sequence $\{P^{(3)}_{N}\}_N$ satisfies the requirements of Proposition \ref{prop:partialconv}, which concludes the proof.
\end{proof}

\section{High-temperature limit} \label{sec:HT}
In this section we apply the results of Section \ref{sec:connection} to study the limit as $\beta \t0$ (the high-temperature limit) of the minimizers of $\fbeta$. We prove their convergence to the law of the Poisson point process in all cases \eqref{wlog}, \eqref{wlog2d}, \eqref{kernel} as stated in Theorem \ref{theo:poisson}.

\subsection{Specific Pinsker inequality} \label{sec:Pinsker}
The well-known Pinsker inequality gives an upper-bound on the total variation distance between probability measures in terms of their Kullback-Leibler divergence:
\begin{equation} \label{classicalPinsker}
 |P-Q|_{\TV} \leq \sqrt{\hal \Ent(P||Q)}
\end{equation}
where $|P-Q|_{\TV}$ is the total variation defined by 
\begin{equation} \label{def:TV}
|P-Q|_{\TV} := \sup\{P(A) - Q(A), A \text{ measurable}\}.
\end{equation}

Combining the Pinsker inequality with the property \eqref{superadditive} of the specific relative entropy we get for any stationary random point process $P$ the following specific (infinite-volume) Pinsker inequality:
\begin{equation} \label{SPI}
\sup_{N \geq 1} \frac{|\Pst_{|\carr_N} - \Poisson_{|\carr_N}|_{\TV}}{|\carr_N|^{\hal}} \leq \sqrt{\frac{1}{2}\ERS[\Pst|\Poisson]}.
\end{equation}

Since the total variation convergence implies weak convergence of probability measures, it is clear that the convergence in specific relative entropy sense implies the weak convergence of random point processes i.e. if a sequence of stationary random point processes $\{P_k\}_k$ satisfies 
$$\lim_{k \ti} \ERS[P_k|\Poisson]=0$$
then the sequence $\{P_k\}_k$ converges to $\Poisson$.

\subsection{Finite energy approximation the Poisson point process} \label{sec:Poissonfroid}
Since the two-point correlation function of the Poisson point process satisfies $\rho_{2,\Poisson} \equiv 1$ we clearly have $\Wint(\Poisson) = 0$ in all cases \eqref{wlog}, \eqref{wlog2d} and \eqref{kernel}. In the case \eqref{kernel} it thus follows that $\Welec(\Poisson)$ is finite, according to the electric-intrinsic inequality of Section \ref{sec:UB}. For the one-dimensional Log-gas it has been proven in \cite{LebSer} that $\Welec(\Poisson) = + \infty$, and the answer is unknown in the two-dimensional Log-gas case (the result of Section \ref{sec:UB} is not enough because $\Dlog(\Poisson)$ is infinite). However we may always construct random point processes which converge in entropy sense to $\Poisson$ and whose renormalized energies are finite.
\begin{lem} \label{lem:Poissonfroid}
There exists a sequence $\{\pi_k\}_k$ of stationary random point processes in $\probas_{s,1}(\config)$ satisfying
\begin{enumerate}
\item $\Welec(\pi_k)$ is finite for all $k$.
\item $\lim_{k \ti} \ERS[\pi_k|\Poisson] = 0$.
\end{enumerate}
\end{lem}
\begin{proof} For any $k \geq 1$, let $\{C_k^i\}_{i \in I}$ be a tiling of $\Rd$ by a countable family of disjoint copies of the hypercube $\carr_k$, and let $\{\B^i_k\}_i$ be the law of a family of independent Bernoulli point processes with $k^d$ points in $C_k^i$. We let $\pi^{t}_k$ be the random point process consisting of the union of all $\B^i_k$ for $i \in I$. Finally we define $\pi_k$ by averaging $\pi^{t}_k$ over a “fundamental domain” i.e. we let
\begin{equation}
\label{def:pik} \pi_k := u_k + \pi^{t}_k,
\end{equation}
where $u_k$ is a uniform random variable in $\carr_k$ (if $\C$ is a point configuration we let $x + \C$ denote the point configuration $\{x + p, p \in \C\}$, cf. also \eqref{def:PZ}).
 
The random point processes $\pi_k$ defined this way are clearly stationary and of intensity $1$. The two-point correlation function of $\pi^t_k$ is easy to compute:
\begin{equation} \label{rho2pit}
\rho^t_2(x,y) = \begin{cases} 1 - \frac{1}{k^d} \text{ if $x$ and $y$ belong to the same hypercube $C^k_i$,} \\ 1 \text{ otherwise.} \end{cases}
\end{equation}
The two-point correlation function of $\pi_k$ could be deduced from \eqref{rho2pit} by averaging $\rho^t_2(x,y)$ over translations of both coordinates by a vector in $\carr_k$. Let us  simply observe that $\rho_2(x,y)-1$ is bounded (because $\rho^t_2$ is) and has compact support (e.g. $\rho_2(x,y) = 1$ as soon as $|x-y| \geq \sqrt{d} k$) which implies (using the expression \eqref{def:Wint2}) that $\Wint(\pi_k)$ is finite. Moreover, observing that $$\int_{\carr_R^2} (\rho^t_2(x,y) -1) = -1$$
we also get that 
$$
\int_{\carr_R^2} (\rho_2(x,y) -1) = -1 + O(R^{d-1})
$$
which implies in view of \eqref{discrasrho2} that $\Esp_{\pi_k}[\D_R^2] = O(R^{d-1})$, hence $\Dlog(\pi_k)$ is zero. Using the electric-intrinsic inequality we conclude that $\Welec(\pi_k)$ is finite for all $k \geq 1$ which proves the first point.

We are left to prove the second point of the lemma. Let $\E^t_k$ be the measurable subset of point configurations which have exactly $k^d$ points in each hypercube $C^i_k$, and let $\E_k$ be the subset obtained by averaging $\E^t_k$ over $\carr_k$, more precisely we let 
\begin{equation} \label{def:Ek}
\E_k := \bigcup_{x \in \carr_k, \C \in \E^t_k } x + \C
\end{equation}
where the sum $x + \C$ is defined as above. By definition $\pi_k$ coincides with the law of the Poisson point process conditioned to the event $\E_k$.
For any $R > 0$ we may thus estimate the relative entropy 
\begin{equation} \label{Entpik}
\Ent\left[ (\pi_k)_{\carr_R} | \Poisson_{|\carr_R} \right] = - \log \Poisson_{|\carr_R}(\E_k).
\end{equation}
Since $\E^t_k \subset \E_k$ we may bound below $\Poisson_{|\carr_R}(\E_k)$ by  $\Poisson_{|\carr_R}(\E^t_k)$ which is easier to compute, indeed we only have to estimate the probability that $$N_{R,k} = \llceil \frac{R}{k} \rrceil^d \approx \frac{R^d}{k^d}$$ disjoint hypercubes of sidelength $k$ receive exactly $k^d$ points, and we can bound below $\Poisson_{|\carr_R}(\E^t_k)$ by
$$
\Poisson_{|\carr_R}(\E^t_k) \geq \left(e^{-k^d} \frac{(k^{d})^{k^d}}{(k^d)!}\right)^{N_{R,k} }.
$$
An elementary estimate using Stirling's formula shows that 
\begin{equation} \label{minorPoissoncond}
- \log \Poisson_{|\carr_R}(\E^t_k) \leq C \frac{R^d}{k^d}
\end{equation}
with a universal constant $C$. We deduce from \eqref{Entpik} and \eqref{minorPoissoncond} that 
$$
\frac{1}{R^d} \Ent\left[ (\pi_k)_{\carr_R} | \Poisson_{|\carr_R} \right] = O(k^{-d})
$$
hence by definition of $\ERS[\cdot | \Poisson]$ we also have
$\ERS[\pi_k | \Poisson] = O(k^{-d}),$ which proves the second point of the lemma.
\end{proof}

\subsection{Proof of Theorem \ref{theo:poisson}} \label{sec:preuveHT}
From Lemma \ref{lem:Poissonfroid} the proof of Theorem \ref{theo:poisson} is straightforward.
\begin{proof} For any $\beta > 0$, let $\Pbeta$ be a minimizer of $\fbeta$. In particular we have
$$
\beta \Welec(\Pbeta) + \ERS[\Pbeta|\Poisson] \leq \beta \Welec(\pi_k) + \ERS[\pi_k|\Poisson] 
$$
for any $k \geq 1$, where $\{\pi_k\}_k$ is the sequence of random point processes constructed in Lemma \ref{lem:Poissonfroid}. Since $\Welec$ is bounded below by some constant depending only on $d,s$ we have
\begin{equation} \label{majorERSpbeta}
\sup_{\fbeta(\Pbeta) = \min \fbeta} \ERS[\Pbeta|\Poisson]  \leq \ERS[\pi_k|\Poisson] + \beta \left(\Welec(\pi_k) - \min \Welec\right).
\end{equation}
Since $\ERS[\pi_k|\Poisson] = o_k(1)$ and $\Welec(\pi_k)$ is always finite, we get \eqref{toPoisson} by considering $k$ large enough and $\beta$ small enough (depending on $k$). The fact that convergence in entropy sense implies weak convergence was observed in Section \ref{sec:Pinsker}.
\end{proof}

Since $\sineb$ was proven to be a minimizer of $\fbeta$ for the one-dimensional Log-gas, we get Corollary \ref{coro:sineb} as an immediate consequence of Theorem \ref{theo:poisson}. This convergence result was recently established in \cite{AllezDumaz} by analysing the family of coupled diffusion processes defining the point processes $\sineb$. Here we rely only on the fact that the $\sineb$ process minimizes the free energy functional $\fbeta$.

\section{Low temperature limit in one dimension} \label{sec:LT}
In this section we prove Theorem \ref{theo:minim1d} i.e. we use the link between $\Welec$ and $\Wint$ to give a minimization result on the energy in the one-dimensional case. As can be expected the minimizer of $\Welec$ is attained by a “crystalline state” which in dimension $1$ corresponds simply to the lattice $\Z$. In the remaining of this section we deal with the cases \eqref{wlog} or \eqref{kernel} with $d=1$ and $0 < s < 1$.

In \cite{SS1d} (in the one-dimensional logarithmic case) the minimality of $\mathbb{W}(\Z)$ among the energies of periodic configurations was proven using an explicit formula valid in the periodic setting, together with the convexity of the interaction kernel. An argument of approximation by periodic configurations was then used to prove that $\Z$ is a global minimizer of the energy (however, it is not unique). In \cite{Leble1d} we turned this convexity argument into a quantitative estimate in order to bound below the difference $\Welec(P) - \Welec(P_{\Z})$ in terms of the two-point correlation function of $P$, first in the periodic case, then in the general stationary case using the same kind of approximation. It was enough to prove that $P_{\Z}$ is the unique minimizer of $\Welec$ among stationary point processes in the case $d=1, s=0$. It also yields the fact that if $\Welec(P_n) \to \Welec(P_{\Z})$ then the two-point correlation function of $P_n$ converges to that of $P_{\Z}$. In the following we prove the same result in all one-dimensional cases, first at the level of \textit{hyperuniform} random point processes (which include periodic point processes) then in the general stationary case using the approximation argument of Proposition \ref{prop:partialconv}. We also observe that convergence of the two-point correlation functions to that of $P_{\Z}$ in fact implies weak convergence of the random point processes.

\subsection{The k-th neighbor correlation functions}
In the one-dimensional case the two-point correlation function of a stationary random point process admits a decomposition as the sum of the $k$-th neighbor correlation functions.

Let $P$ be in $\probas_{s,1}(\config)$ such that the two-point correlation $\rho_{2}$ exists as a Radon measure in $\R \times \R$. For any $k \geq 1$ we define the $k$-th neighbor correlation function $\rho_{2,k}$ by duality, letting for any $\varphi \in C_c (\R \times \R)$
\begin{equation}
\label{def:rho2k} \int  \varphi \rho_{2,k} := \frac{1}{2} \Esp_P \left[ \sum_{x,y \in \C, y \text{ $k$-th neighbor of }x} \varphi(x,y) + \varphi(y,x) \right].
\end{equation}
In \eqref{def:rho2k} if $x,y$ belong to a point configuration $\C$ we say that $y$ is the $k$-th neighbor of $x$ if $x < y$ and $\C([x,y]) = k + 1$. We will abbreviate “$k$-th neighbor of” by \textit{k.n.o.} in the formulas. Since $P$ is stationary we may see $\rho_{2,k}$ as a measure on $\R$ by letting $\rho_{2,k}(x) := \rho_{2,k}(0,x)$ (in the rest of this section we will use the same notation for both interpretations of $\rho_{2,k}$). 
 
\begin{lem} \label{lem:propRho2k} For any $k \geq 1$, for any compactly supported, measurable, even function $\varphi : \R \mapsto \R$, we have
\begin{equation} \label{rho2kvsEsp}
\Esp_P \left[ \sum_{x,y \in \C \cap [-L/2, L/2], y\ k.n.o.\ x} \varphi(x-y) \right] = \int_{0}^{L} \varphi(x) \rho_{2,k}(x) (1 - x/L).
\end{equation}
\end{lem}
\begin{proof}
We use the definition \eqref{def:rho2k} together with a change of variable $(x,y) \mapsto (x-y, x+y)$ as in the re-writing of $\Wint$ in Section \ref{sec:Wint}.
\end{proof}

\subsection{Minimization: the hyperuniform case}

\begin{lem} \label{lem:preMinim} The unique minimizer of $\Wint$ among random point processes in $\probas_{\rm{hyp}}(\config)$ is the random point process $P_{\Z}$ defined in \eqref{def:PZ}. Moreover for any such $\Pst$ we have
\begin{equation} \label{controlDW}
\Wint(\Pst) - \Wint(P_{\Z}) \geq c \sum_{k=1}^{+ \infty} \int_{0}^{+ \infty} \min\left( \frac{(x-k)^2}{k^{s+2}}, 1 \right) \rho_{2,k}(x),
\end{equation}
with a constant $c$ depending only on $s$.
\end{lem}
\begin{proof} 
Since there are a.s. $r$ points in any interval of length $r$ the $k$-th neighbor correlation function of $P$ is supported in $[k-r, k+r]$. We may thus write, for any $R > 0$
$$
\frac{1}{R} \int_{-R}^R g(x) (\rho_2(x) -1) (R-|x|) dx = \frac{2}{R} \int_{0}^R g(x) \left(\sum_{k=1}^{R+r} \rho_{2,k} - 1\right) (R-x) dx.
$$
Let $\psi_R : x \mapsto \frac{2}{R} g(x) (R-x)$. By definition of $\Wint$ (see \eqref{def:Wint2}) we get
\begin{equation} \label{formuleP}
\Wint(P) = \liminf_{R \ti} \int_0^R \psi_R(x) \left(\sum_{k=1}^{R+r} \rho_{2,k} - 1\right) dx,
\end{equation}
and for $P_{\Z}$ it is easy to see that the $\liminf$ is actually a $\lim$ and we have
\begin{equation} \label{formulePZ}
\Wint(P_{\Z}) = \lim_{R \ti} \int_0^R \psi_R(x) \left(\sum_{k=1}^{\floor{R}} \delta_k - 1\right) dx.
\end{equation}
Substracting \eqref{formulePZ} in \eqref{formuleP} we get
$$
\Wint(P) - \Wint(P_{\Z}) = \liminf_{R \ti} \int_0^R \psi_R(x) \left( \sum_{k=1}^{R+r} \rho_{2,k}  - \sum_{k=1}^{\floor{R}} \delta_k\right).
$$
We may re-write the previous expression as
\begin{equation} \label{rewriteDW}
\Wint(P) - \Wint(P_{\Z}) = \liminf_{R \ti} \int_{0}^{+\infty} \psi_R(x) \left( \sum_{k=1}^{\floor{R-r}} (\rho_{2,k} - \delta_k) \right) + E_{R,r}
\end{equation}
where the error term $E_{R,r}$ is bounded using the fact that $\rho_{2,k}$ is supported on $[k-r, k+r]$ and that $|\psi_R|$ is decreasing on $[R-2r, R]$ (for $R$ large enough).
\begin{equation} \label{controleERr}
|E_{R,r}| \leq \int_0^R  \psi_R(x)\left( \sum_{k=\floor{R-r}}^{R+r} \rho_{2,k} + \sum_{k=\floor{R-r}}^{\floor{R}} \delta_k \right) \leq Cr |\psi(R-2r)| = o_R(1).
\end{equation}
Let us now observe that $\psi_R$ is a convex function, more precisely for $x \in (0, +\infty)$ we have $\psi_R''(x) \geq \frac{c}{x^{s+2}}$ for some positive constant $c$ depending on $s$. Moreover  for all $k \geq 1$, since $P$ is of intensity $1$ we have $\int \rho_{2,k} = 1$ ($\rho_{2,k}$ is the probability law of the $k$-th neighbor) and since $P$ is periodic the expectation $\int x \rho_{2,k}(x)$ is finite and thus equal to $k$ (the $k$-th neighbor is in average at distance $k$). Combining this observation with the convexity estimate we may write
$$
\int_0^{+\infty} \psi_R(x) (\rho_{2,k} - \delta_k) \geq \int_0^{+\infty} \frac{(x-k)^2}{\max(x,k)^{s+2}} \rho_{2,k}(x) \geq c \int_0^{+\infty} \min \left( \frac{(x-k)^2}{k^{s+2}}, 1 \right) \rho_{2,k}(x),
$$
with $c$ depending only on $s$. Inserting this bound in \eqref{rewriteDW} and using \eqref{controleERr} we get
$$
\Wint(P) - \Wint(P_{\Z}) \geq \liminf_{R \ti} \sum_{k=1}^{\floor{R-r}} \int_0^{+\infty} \min \left( \frac{(x-k)^2}{k^{s+2}}, 1 \right) \rho_{2,k}(x) + o_R(1)$$
Hence finally by taking the limit $R \ti$ we obtain that
\begin{equation}
\Wint(P) - \Wint(P_{\Z}) \geq \sum_{k=1}^{+ \infty} \int_0^{+\infty} \min \left( \frac{(x-k)^2}{k^{s+2}}, 1 \right) \rho_{2,k}(x),
\end{equation}
which proves \eqref{controlDW}. This lower bound implies that $P_{\Z}$ is a minimizer of $\Wint$ among hyperuniform random point processes and also that it is unique. Indeed $\Wint(P) = \Wint(P_{\Z})$ implies (by \eqref{controlDW}) that $\rho_{2,k} = \delta_{k}$ for all $k \geq 1$, hence the two-point correlation function of $P$ coincides with the one of $P_{\Z}$, which is enough to conclude that $P = P_{\Z}$ (see e.g. \cite{Leble1d}).
\end{proof}

\subsection{Proof of Theorem \ref{theo:minim1d}}
We may now give the proof of our minimization result for $\Welec$. First let us observe that $\Welec(P_{\Z}) = \Wint(P_{\Z})$. The inequality $\leq$ is true by Proposition \ref{prop:elecint}. Moreover Proposition \ref{prop:partialconv} ensures that there exists a sequence of \textit{hyperuniform} random point processes such that $\limsup_{N \ti} \Wint(P_N) \leq \Welec(P_{\Z})$. By Lemma \ref{lem:preMinim} above we know that $\Wint(P_N)  \geq \Wint(P_{\Z})$, hence in fact $\Wint(P_{\Z}) \leq \Welec(P_{\Z})$ and equality holds.

\begin{proof}
\textbf{Step 1.} \textit{Minimization of $\Welec$.}
 Let $P$ be a minimizer of $\Welec$ on $\probas_{s,1}(\config)$. From Proposition \ref{prop:partialconv} we get a sequence $\{P_n\}_n$ of hyperuniform random point processes converging to $P$ and such that $\{\Wint(P_n)\}_n$ converges to $\Welec(P)$. For any $k,n \geq 1$ let $\rho^{(n)}_{2,k}$ denote the $k$-th neighbour correlation function of $P_n$. In the hyperuniform case, \eqref{controlDW}  implies that for any $M > 0$ and $k \geq 1$ we have
\begin{equation} \label{applicationDW}
 \Wint(P_n) - \Wint(P_{\Z}) \geq \int_0^{2M} \min \left( \frac{(x-k)^2}{k^{s+2}}, 1 \right) \rho^{(n)}_{2,k}(x), 
\end{equation}
and the right-hand side is bounded below by (see \eqref{rho2kvsEsp})
$$
\int_0^{2M} \min \left( \frac{(x-k)^2}{k^{s+2}}, 1 \right) \rho^{(n)}_{2,k}(x) \left(1-\frac{x}{2M}\right) = \Esp_{P_n} \left[ \sum_{x,y \in \C \cap [-M, M], y\ k.n.o.\ x} \varphi_k(x-y) \right] 
$$
where we let $\varphi_k(x) := \min \left( \frac{(x-k)^2}{k^{s+2}}, 1 \right)$. Since $\{P_n\}_n$ converges to $P$ we have
\begin{multline} \label{rho2Pntorho2p}
\Esp_{P_n} \left[ \sum_{\C \cap [-M, M], y\ k.n.o.\ x} \varphi_k(x-y) \right] = \Esp_{P} \left[ \sum_{\C \cap [-M, M], y\ k.n.o.\ x} \varphi_k(x-y) \right] + o_n(1) 
\\ = \int_0^{2M} \varphi_k(x) \rho_{2,k}(x) \left(1-\frac{x}{2M}\right) + o_n(1) \geq \frac{1}{2} \int_0^{M} \varphi_k \rho_{2,k}(x) + o_n(1).
\end{multline}
Combining \eqref{applicationDW} and \eqref{rho2Pntorho2p} we see that
$$
\int_0^{M} \min \left( \frac{(x-k)^2}{k^{s+2}}, 1 \right) \rho_{2,k}(x) \leq 2 \left(\Wint(P_n) - \Wint(P_{\Z})\right) + o_n(1),
$$
but as $n \ti$ we have $\Wint(P_n) \to \Welec(P) \leq \Welec(P_{\Z}) = \Wint(P_{\Z})$. It implies that 
$$
\int_0^{M} \min \left( \frac{(x-k)^2}{k^{s+2}}, 1 \right) \rho_{2,k}(x)  = 0
$$
for all $M > 0$ and $k \geq 1$. Finally we get that $\rho_{2,k} = \delta_k$ for all $k \geq 1$ and we conclude as in the proof of Lemma \ref{lem:preMinim} that $P = P_{\Z}$, which ensures that $P_{\Z}$ is the unique minimizer of $\Welec$ on $\probas_{s,1}(\config)$.

If $P$ is not a minimizer the same argument shows that
\begin{equation} \label{DeltaWgeneral}
\Welec(P) - \Welec(P_{\Z}) \geq \sum_{k=1}^{+\infty} \int_0^{+\infty} \min \left( \frac{(x-k)^2}{k^{s+2}}, 1 \right) \rho_{2,k}(x)
\end{equation}
as in the hyperuniform setting. 

\textbf{Step 2.} \textit{Energy of minimizers of $\fbeta$ tends to $\Welec(P_{\Z})$.}
On the other hand we claim that if $\{P_{\beta}\}_{\beta > 0}$ is a family of minimizers of $\fbeta$ then we must have
\begin{equation} \label{convMinPbeta}
\lim_{\beta \t0} \Welec(P_{\beta}) = \Welec(P_{\Z}).
\end{equation}
To prove \eqref{convMinPbeta} we cannot directly evaluate $\fbeta$ over $P_{\Z}$ and compare with $P_{\beta}$ because $\ERS[P_{\Z}|\Poisson]$ is infinite. However we may argue as in Section \ref{sec:Poissonfroid} and show that there exists a sequence $\{\pi_{k}\}_k$ of stationary random point processes in $\probas_{s,1}(\config)$ satisfying 
\begin{enumerate}
\item $\ERS[\pi_k|\Poisson]$ is finite for all $k$.
\item $\lim_{k \ti} \Welec[\pi_k] = \Welec(P_{\Z})$.
\end{enumerate}
Such a sequence can be constructed by chosing a “vibrating” approximation of $P_{\Z}$. For any $k \geq 1$ we let $\{V_{k,m}\}_{m \in \Z}$ be a countable family of i.i.d. random variables distributed uniformly in $[-\frac{1}{k}, \frac{1}{k}]$, then we let $\pi^t_{k}$ be the random point process 
$$\pi^t_k := \sum_{m \in \Z} \delta_{m + V_{k,m}}$$
and finally we define $\pi_k$ by averaging $\pi^t_k$ over $[0,1]$. It is easy to check that $\pi_k$ is a stationary random point process of intensity $1$. In fact $\pi_k$ may equivalently be defined as a renewal process with increments distributed as $1 + V_{k,2} - V_{k,1}$. The specific relative entropy of $\pi_k$ coincides with its “entropy rate”, it is finite (see \cite[Section 14.8.]{dvj2}) and blows up as $k \ti$ like the entropy of $V_{k,2} - V_{k,1}$. Concerning the energy, we have the bound
$$
\Welec(P_{\Z}) \leq \Welec(\pi_k) \leq \Wint(\pi_k)
$$ 
and the fact that $\Wint(\pi_k)$ converges to $\Wint(P_{\Z}) = \Welec(P_{\Z})$ can be checked directly with the help of the formula defining $\Wint$. Indeed the two-point correlation function of $\pi_k$ may be written as
\[
\rho_{2,\pi_k} = \sum_{m \in \Z} \psi_k(m + \cdot)
\]
where $\psi_k$ is a triangular “hat function” of width $\frac{1}{2k}$ and integral $1$. For any $m \in \Z$ and $R > 0$, a mean value argument shows that
\[
\left|\int \psi_k(m + \cdot) \log|x|(1-\frac{|x|}{R}) - \log|m|(1 - \frac{m}{R})\right| \leq \frac{C}{k^2} \left( \frac{1}{m^2} + \frac{1}{mR} \right).
\]
Consequently, we get
\[
\limsup_{R \ti} \left| \int_{[-R,R]} \left(\rho_{2,\pi_k} - \rho_{2,\Z}\right)(1- \frac{|x|}{R}) \right| = O\left(\frac{1}{k^2}\right),
\]
and we obtain that $\lim_{k \ti} \Wint(\pi_k)  = \Wint(P_{\Z})$.

With the help of the sequence $\{\pi_k\}_k$ we obtain \eqref{convMinPbeta} by arguing as in Section \ref{sec:preuveHT}, 
indeed since $\ERS[P_{\beta}|\Poisson]$ is always non-negative we have
$$
\Welec(P_{\beta}) \leq \Welec(\pi_k) + \frac{1}{\beta} \ERS[\pi_k|\Poisson] \leq \Welec(P_{\Z})+ \frac{1}{\beta} \ERS[\pi_k|\Poisson] + o_k(1)
$$
and \eqref{convMinPbeta} follows by chosing $k, \beta$ large enough.

\textbf{Step 3.} \textit{Convergence of the two-point function of minimizers of $\fbeta$.}
Combining \eqref{DeltaWgeneral} and \eqref{convMinPbeta} we see that if $\rho^{(\beta)}_{2,k}$ denotes the $k$-th neighbor correlation function of $P_{\beta}$, we have for any 
\[x
\sum_{k=1}^{+\infty} \int_0^{+\infty} \min \left( \frac{(x-k)^2}{k^{s+2}}, 1 \right) \rho^{(\beta)}_{2,k}(x) \rightarrow 0  
\]
as $\beta \ti$. Arguing as in the proof of \cite[Lemma 2.3]{Leble1d} we deduce that $\rho^{(\beta)}_2$ converges to $\sum_{k \in \Z^{*}} \delta_{k}$ in the distributional sense as $\beta \to \infty$. Let us observe that
\[
\rho_{2,\Z} := \sum_{k \in \Z^{*}} \delta_{k}
\]
is the two-point correlation function of $P_{\Z}$.

\textbf{Step 4.} \textit{Weak convergence of the minimizers of $\fbeta$.}
It is not hard to see that this convergence implies in fact the weak convergence of $P_{\beta}$ to $P_{\Z}$ as $\beta \to \infty$. For any $\hal > \epsilon > 0$ let $\chi_{\epsilon}$ be a smooth non-negative function which is equal to $1$ on the set $\cup_{k \in \Z} [k-1 + \epsilon, k- \epsilon]$ and to $0$ on $\Z$. For any $T > 0$ we let $\tau_{T,\epsilon}$ be a non-negative continuous function such that $\tau_{T,\epsilon} \equiv 1$ on $[-T+\ep, T-\ep]$ and $0$ outside $[-T, T]$. We let $\varphi_{\epsilon, T}$ be the continuous, compactly supported map
\[
\varphi_{\epsilon, T}(x,y) := \chi_{\epsilon}(x-y) \tau_{T, \epsilon}(x) \tau_{T,\epsilon}(y).
\]
 Let $A_{T,\epsilon}$ be the event “there is no pair $(x,y)$ of points of the configuration in $[-T+\epsilon,T-\epsilon]$ such that $|x-y| \in \cup_{k \in \Z} [k-1 + \epsilon, k-\epsilon]$”. Since $\int \varphi_{\epsilon, T} \rho_{2,\Z} = 0$ and since $\rho^{(\beta)}_2$ converges to $\rho_{2,\Z}$ as $\beta \to \infty$ we have 
\begin{equation} \label{patep}
P_{\beta}\left(A_{T,\epsilon}\right) \longrightarrow 1
\end{equation}
as $\beta \to \infty$. In other words, with probability tending to $1$ as $\beta \to \infty$, a configuration under $P_{\beta}$ locally looks like a (translated) subset of $\Z$ in which all the points have been displaced at a distance at most $\epsilon$.

The variance under $P_{\beta}$ of the number of points in $[-T,T]$ is bounded as $\beta \ti$, because it is controlled by the energy $\Welec(P_{\beta})$, which itself converges (to $\Welec(P_{\Z})$). This follows from the discrepancy estimates (see e.g. \cite[Lemma 2.1]{Leble1d} or \cite[Lemma 3.10]{LebSer}). In particular we have uniform integrability under $P_{\beta}$ of the number of points in $[-T,T]$ as $\beta \to \infty$. In particular, conditioning $P_{\beta}$ to $A_{T,\epsilon}$, we have an average of $2T + o(1)$ points in $[-T,T]$ as $\beta \ti$, because \eqref{patep} holds.

Finally we deduce that for any $\epsilon > 0$, with probability tending to $1$ as $\beta \to \infty$, a configuration under $P_{\beta}$ locally looks like a translate of $\Z$ in which all the points have been displaced at a distance at most $\epsilon$. This implies the convergence of $P_{\beta}$ to $P_{\Z}$ as $\beta \ti$.

\end{proof}

\bibliographystyle{alpha}
\bibliography{WBSbib}
\end{document}